\documentclass[11pt,letterpaper]{article}
\usepackage[utf8]{inputenc}
\usepackage{times}
\usepackage{url}
\usepackage{fullpage}
\usepackage{pslatex} 
\usepackage{mathdots} 
\usepackage{amsmath}
\usepackage{algorithm}
\usepackage[noend]{algpseudocode}
\usepackage{amssymb}
\usepackage{amsthm}
\usepackage{color}
\usepackage{array}
\usepackage{xy}
\usepackage{setspace}
\usepackage{multicol}
\usepackage{hyperref}
\usepackage{soul}
\usepackage{cite}
\usepackage[margin=1in,letterpaper]{geometry}
\renewcommand{\log}{\lg}
\usepackage{accents}

\newcommand{\defn}[1]{\emph{\textbf{{#1}}}}

\usepackage{todonotes}

\newcommand{\poly}{\operatorname{poly}}
\newcommand{\polylog}{\operatorname{polylog}}
\newcommand{\E}{\mathbb{E}}

\newcommand{\OR}{\vee}

\newcommand{\AND}{\wedge}

\renewcommand{\paragraph}[1]{\vspace{.5 cm} \noindent \textbf{#1} }
\usepackage{thmtools}
\usepackage{thm-restate}

\newtheoremstyle{slanted}
{3pt}
{3pt}
{\slshape}
{}
{\bfseries}
{.}
{.5em}
{}
\theoremstyle{slanted}
\newtheorem{theorem}{Theorem}
\newtheorem{lemma}[theorem]{Lemma}

\newtheorem{proposition}[theorem]{Proposition}

\newtheorem{corollary}[theorem]{Corollary}
\newtheorem*{remark}{Remark}

\title{Stochastic and Worst-Case Generalized Sorting Revisited}
\author{William Kuszmaul, Shyam Narayanan}
\date{MIT}

\begin{document}
\maketitle

\begin{abstract}
    The \emph{generalized sorting problem} is a restricted version of standard comparison sorting where we wish to sort $n$ elements but only a subset of pairs are allowed to be compared. Formally, there is some known graph $G = (V, E)$ on the $n$ elements $v_1, \dots, v_n$, and the goal is to determine the true order of the elements using as few comparisons as possible, where all comparisons $(v_i, v_j)$ must be edges in $E$. We are promised that if the true ordering is $x_1 < x_2 < \cdots < x_n$ for $\{x_i\}$ an unknown permutation of the vertices $\{v_i\}$, then $(x_i, x_{i+1}) \in E$ for all $i$: this Hamiltonian path ensures that sorting is actually possible.
    
    In this work, we improve the bounds for generalized sorting on both random graphs and worst-case graphs. For Erd\H{o}s-Renyi random graphs $G(n, p)$ (with the promised Hamiltonian path added to ensure sorting is possible), we provide an algorithm for generalized sorting with an expected $O(n \log (np))$ comparisons, which we prove to be optimal for query complexity. This strongly improves over the best known algorithm of Huang, Kannan, and Khanna (FOCS 2011), which uses $\tilde{O}(\min(n \sqrt{np}, n/p^2))$ comparisons. For arbitrary graphs $G$ with $n$ vertices and $m$ edges (again with the promised Hamiltonian path), we provide an algorithm for generalized sorting with $\tilde{O}(\sqrt{mn})$ comparisons. This improves over the best known algorithm of Huang et al., which uses $\min(m, \tilde{O}(n^{3/2}))$ comparisons.
\end{abstract}

\section{Introduction} \label{sec:intro}

Whereas the standard comparison-based sorting problem was solved more than half a century ago,
several variations on the problem have only recently begun to be understood.
One of the most natural of these is the so-called generalized sorting problem.

\paragraph{Generalized sorting: sorting with forbidden comparisons.}
In the generalized sorting problem, we wish to sort the vertices of a graph, 
but only certain pairs of vertices are permitted to be compared. The input to the problem is a graph $G = (V, E)$,
where the vertices represent the elements to be sorted, and the edges indicate the pairs of vertices
that can be compared; we will use $x_1, x_2, \ldots, x_n$ to denote the true ordering of the vertices $V$,
and we say that $x_i \prec x_j$ if $i < j$. The goal is to sort the vertices (i.e., discover the true ordering)
with as few comparisons as possible.
In this setting, comparisons are also referred to as \defn{edge queries},
meaning that the goal is to achieve
the minimum possible query complexity.

In order for an input to the generalized sorting problem to be valid, it is required that the true ordering $x_1 \prec x_2 \prec \cdots\prec x_n$ of the vertices $V$ appears as a Hamiltonian path in $G$. This ensures that, should an algorithm query every edge,
the algorithm will be able to deduce the true order of the vertices.

A classic special case of the generalized sorting problem is the nuts-and-bolts problem,
which considers the case where $G$ is a complete bipartite graph. The vertices on one side of the graph represent ``nuts'' and vertices on the other side of the graph represent
``bolts''. The goal is to compare nuts to bolts in order to match up each nut with its correspondingly sized bolt.
Nuts cannot be compared to other nuts and bolts cannot be compared to other bolts, which is why the graph has a complete bipartite structure.
The problem was first introduced in 1994 by Alon et al. \cite{alon1994matching}, who gave a {deterministic} $O(n \log^4 n)$-time algorithm. Subsequent work 
\cite{alon1996matching, komlos1998matching, bradford1995matching} improved the running time to $O(n \log n)$, which is asymptotically optimal.

The first nontrivial bounds for the full generalized sorting problem were given by Huang et al. \cite{huang2011algorithms}, who studied two versions of the problem:
\begin{itemize} 
\item \textbf{Worst-case generalized sorting: } In the worst-case version of the problem, $G$ is permitted to have an arbitrary structure.
A priori, it is unclear whether it is even possible to achieve $o(n^2)$ comparisons in this setting. The paper \cite{huang2011algorithms} showed that this is always possible,
{with a randomized algorithm}
achieving query complexity $\tilde{O}(n^{1.5})$ for any graph $G$, with high probability. 
\item \textbf{Stochastic generalized sorting: } In the stochastic version of the problem, every vertex pair $(u, v)$ (that is not part of the Hamiltonian path for the true order)
is included independently and randomly with some probability $p$. The paper \cite{huang2011algorithms} gave a randomized algorithm with $\tilde{O}(\min\{n / p^2, n^{3/2} \sqrt{p}\})$ query complexity,
which evaluates to $\tilde{O}(n^{1.4})$ for the worst-case choice of $p$. 
\end{itemize}

Although there have been several subsequent papers on the topic \cite{banerjee2016sorting, lu2021generalized, biswas2017improved} (which we will discuss further shortly), 
these upper bounds of $\tilde{O}(n^{1.5})$ and $\tilde{O}(n^{1.4})$ have remained the state of the art for the worst-case and stochastic generalized sorting problems, respectively.
Moreover, no nontrivial lower bounds are known. Thus it is an open question how these problems compare to standard comparison-based sorting.

\paragraph{This paper. }
The main result of this paper is an algorithm for stochastic generalized sorting with query complexity $O(n \log (pn))$.
This means that, in the worst case, the query complexity of stochastic generalized sorting is asymptotically the same as that of classical sorting.
Perhaps even more remarkably, when $p$ is small, stochastic generalized sorting is actually \emph{easier} than its classical counterpart; for example,
when $p = \frac{\polylog(n)}{n}$, the query complexity becomes $O(n \log \log n)$.

When $p = (\log n +\omega (1))/n $, which is the parameter regime in which the Erd\"os-Renyi random graph $G(n, p)$ contains at least one Hamiltonian cycle with overwhelming probability,
we prove a matching lower bound of $\Omega (n\log (pn)) $ for the query complexity of any stochastic-generalized-sorting algorithm. Thus our algorithm is optimal.


Given that the optimal running time for stochastic generalized sorting is faster for sparser graphs, a natural question is whether the running time for worst-case generalized sorting can also be improved in the sparse case. It is known that
in the \emph{very dense} case, where there are $\binom{ n } { 2 } - q $ edges, there is an algorithm with query complexity $(n + q)\log n $ \cite{banerjee2016sorting}. However, in the case where there are
$m\ll \binom{ n } { 2 }$ edges, the state-of-the-art remains the $\min(m, \tilde{O} (n ^ { 1. 5 }))$ bound of \cite{huang2011algorithms}, where one can always get an $O(m)$-query bound simply by querying all edges.

Our second result is a new algorithm for worst-case generalized sorting with query complexity $\tilde{O} (\sqrt { nm }) $, where $n $ is the number of vertices and $m $ is the number of edges.
The algorithm is obtained by combining the convex-geometric approach of \cite{huang2011algorithms} with the predictions-based approach of \cite{banerjee2016sorting}.

Interestingly, if we instantiate $m = p \binom{n}{2} $ for some $p\in (0, 1] $, then the running time becomes $\tilde{O} (n^{1.5} \sqrt{p})$, which is precisely what the previous state of the art was
was for the \emph{stochastic} generalized sorting problem (in the sparse case of $p < 1 / n^{1/5}$) \cite{huang2011algorithms}.

\paragraph {Other related work. }
The original bounds by Huang et al.\cite{huang2011algorithms} for both the worst-case and the stochastic generalized sorting problems
remained unimproved until now. The difficulty of these problems has led researchers to consider alternative formulations
that are more tractable. Banerjee and Richards \cite{banerjee2016sorting} 
considered the worst-case generalized sorting problem in the setting where $G$ is very dense, containing $\binom{n}{2} - q$ edges for some parameter $q$, 
and gave an algorithm that performs $O((q + n) \log n)$ comparisons; whether this is optimal remains open, and the best known lower bound is $\Omega(q + n \log n)$ \cite{biswas2017improved}. Banerjee and Richards \cite{banerjee2016sorting} also gave on alternative algorithm for the stochastic generalized sorting problem, achieving 
$\tilde{O}(\min\{n^{3/2}, pn^2\})$ comparisons, but this bound is never better than that of \cite{huang2011algorithms} for any $p$.
Work by Lu et al. \cite{lu2021generalized} considered a variation on the worst-case generalized sorting problem in which we are also given
predicted outcomes for all possible comparisons, and all but $w$ of the predictions are guaranteed to be true. The authors \cite{lu2021generalized} show that,
in this case, it is possible to sort with only $O(n \log n + w)$ comparisons. 

The generalized sorting problem is also closely related to the problem of sorting with comparison costs.
In this problem, we are given as input $\binom{n}{2}$ nonnegative costs 
$\{c_{u, v} \mid u, v \in V\}$, where $V$ is the set of elements that must be sorted and $c_{u, v}$ 
is the cost of comparing $u$ to $v$. The goal is to determine the true order $x_1 \prec x_2 \prec \cdots \prec x_n$ 
of $V$ at a cost that is as close as possible to the optimal cost 
$$\text{OPT} = c_{x_1, x_2} + c_{x_2, x_3} + \cdots + c_{x_{n - 1}, x_n}.$$

If all comparisons have equal costs, the problem reduces to traditional sorting and the optimal competitive ratio is $O (\log n) $. 
On the opposite side of the spectrum, if the costs are allowed to be arbitrary, then the best achievable competitive ratio is known to be $\Omega(n)$ \cite{gupta2001sorting}.  
Nonetheless, there are many specific families of cost functions for which better bounds are achievable.
Gupta and Kumar \cite{gupta2001sorting} consider the setting in which each of the elements $v \in V$ represents a database entry of some size $s_v$,
and the costs $c_{u, v}$ are a function of the sizes $s_v$ and $s_u$; two especially natural cases are when 
$c_{u, v} = s_v + s_u$ or $c_{u, v} = s_u s_v$, and in both cases, a competitive ratio of $O(\log n)$
can be achieved \cite{gupta2001sorting}.\footnote{ The authors \cite{gupta2001sorting} also describe an algorithm for arbitrary monotone cost functions,
but the analysis of the algorithm has been observed by subsequent work to be incorrect in certain cases \cite{kannan2003selection}.}
Several other families of cost functions have also been considered. Gupta and Kumar \cite{gupta2005s} consider the setting in which the costs induce a metric space over the elements $V $, and give a $O (\log ^ 2n) $-competitive
algorithm. In another direction, Angelov et al. \cite{angelov2008sorting} consider costs that are drawn independently and uniformly from certain probability distributions
(namely, the uniform distribution on $[0, 1] $, and the Bernoulli distribution with $\Pr [X = 1] = p $), and give a $O (\log n) $-competitive algorithm in each case.
The unifying theme across these works is that, by focusing on specific families of cost functions, polylogarithmic competitive ratios can be achieved.

The generalized sorting problem can also be interpreted in this cost-based framework, with queryable edges having cost 1 and un-queryable edges having cost $\infty$. 
The presence of very large (infinite) costs significantly changes the nature of the problem, however.
The fact that some edges simply cannot be queried makes it easy for an algorithm to get ``stuck'',
unable to query the edge that it needs to query in order to continue. 
Whereas sorting with well-behaved (and bounded) cost functions is, at this point, a relatively well understood problem, 
the problem of sorting with infinite costs on some edges (i.e., generalized sorting) has remained much more open.


\subsection{Roadmap}

In Section \ref{sec:Technical}, we provide a technical overview of our stochastic generalized sorting algorithm and its analysis. This result is our most interesting and technically involved result. In Section \ref{sec:Random}, we provide a full description and analysis of the stochastic generalized sorting algorithm, which uses $O(n \log (np))$ queries on a random graph $G(n, p)$. In Section \ref{sec:Lowerbound}, we prove that this algorithm is essentially optimal on random graphs. In Section \ref{sec:Worstcase}, we provide an improved worst-case generalized sorting algorithm that uses $O(\sqrt{mn} \log n)$ queries on an arbitrary graph with $n$ vertices and $m$ edges.

Finally, in Appendix \ref{sec:Pseudocode}, we provide pseudocode for the algorithms in Section \ref{sec:Random} and Section \ref{sec:Worstcase}.

\section{Technical Overview} \label{sec:Technical}
In this section, we give a technical overview of the main technical result of the paper, which is our
algorithm for stochastic generalizing sorting that runs in $ O (n\log (np)) $ queries.

\paragraph{The setup.}
Suppose that we are given an instance $ G = (V, E) $ of the stochastic generalized sorting problem,
where $|V| = n$, and where $p$ is the edge sampling probability. Let $ x_1, x_2,\ldots, x_n $
denote the true order of the vertices (and recall that we use $ x_i\prec x_j $ to denote that $ i <j $). Note that the edges
$(x_i, x_{i +1})$ are all included deterministically in $E$, and that each other edge is included
independently and randomly with probability $ p $. Further note that the edges in $E$ are undirected, so the edge $(u, v)$ is the same as the edge $(v, u)$.

\paragraph{The basic algorithm design.}
The algorithm starts by constructing sets of edges $ E_1, E_2,\ldots, E_q $, where $ q =\log (np) $ and where each
$ E_i $ consists of roughly a $1 / 2 ^ i $ fraction of all the edges $ E $. These sets of edges
are constructed once at the beginning of the algorithm, and are never modified.

The algorithm then places the vertices into levels $ L_1 \subseteq L_2 \subseteq \cdots \subseteq L_{q + c}$, 
where $ c $ is some large positive constant. The top levels $L_{q + 1}, \ldots, L_{q + c}$ automatically contain
all of the vertices. To construct the lower levels, we use the following \defn{promotion rule} to determine whether
a given vertex $ v $ in a level $ L_{ i +1 } $ should also be placed into level $ L_i $: 
if there is some vertex $u \prec v$ such that $u \in L_{i + c}$ and $(u, v) \in E_i$, 
then we consider $ v $ to be \defn{blocked} by $u$; otherwise, if there is no such $ u $,
then $ v $ gets promoted to level $ L_i $.

As we shall describe shortly, our algorithm discovers the true order of the vertices $ x_1, x_2, \ldots$ one vertex at a time.
At any given moment, we will use $ x_\ell$ to denote the most recently discovered vertex in the true order,
and we will say that each remaining vertex $ x_i $, $ i > \ell$, has \defn{rank} $ r (x_i) = i - \ell$. 

Before we describe how to discover the true order of the vertices, let us first describe
how the levels are updated as we go. Intuitively, the goal of this construction is to make it so that,
at any given moment, each level $ L_i $ has size roughly
$ s_i = 2 ^ i/p $, and so that the majority of the elements in $ L_i $ have rank $O(s_i)$. We refer to
$s_i = 2^i / p$ as the \defn{target size} for level $L_i$. 

Whenever we discover a new vertex $ x_\ell$ (i.e., we discern that some vertex $q$ is the $\ell$-th vertex in the true order), we perform incremental updates to the levels as follows: we remove
$x_\ell$ from all of the levels, and then for each vertex $ v $ that was previously blocked from promotion by
$ x_\ell $, we apply the promotion rule as many times as necessary to determine what new level $v$ should be promoted
to. Note that, when we apply the promotion rule to a vertex $ v $ in a level $L_{i + 1}$, we apply it based on the current set $ L_{ i + c } $,
and even if the set $L_{i + c}$ changes in the future, we do not consequently ``unpromote'' the vertex.

In addition to the incremental updates, we also periodically rebuild entire levels from scratch. In particular, we rebuild
 the levels $L_i, L_{i - 1}, \ldots, L_1$ whenever the number $\ell$ of vertices that we have discovered is a multiple of
$s_i / 32$. (Recall that $s_i = 2^i / p$ is our target size for each level $L_i$.) To rebuild a level $ L_i $, we take the vertices currently in $ L_{ i +1 } $ and apply the promotion rule to each of them (based on the current values of the sets $L_{i+1}$ and $L_{i + c}$). 

We remark that, whereas incremental updates can only promote vertices (and cannot unpromote them), the process of rebuilding the levels $L_i, L_{i - 1}, \ldots, L_1$, one after another, can end up moving vertices in both directions. Moreover, there are interesting chains of cause and effect. For example, if a vertex $v$ gets unpromoted from level $L_i$ (so now it no longer appears in $L_i$), then that may cause some vertex $v' \in L_{i - c + 1}$ to no longer be blocked, meaning that $v'$ now gets \emph{promoted} to $L_{i - c}$; the addition of $v'$ to $L_{i - c}$ may block some $v'' \in L_{i - 2c + 1}$ from being able to reside in level $L_{i - 2c}$, causing $v''$ to get \emph{unpromoted} from $L_{i - 2c}$, etc.

Having described how we maintain the levels $ L_1, L_2,\ldots$, let us describe how we use the levels to discover
the next vertex $ x_{\ell +1 } $. We first construct a \defn{candidate set} $ C $ consisting of the vertices
$ v\in L_1$ such that $(x_\ell, v) \in E$. Then, for each $v \in C$, we determine whether $v = x_{\ell + 1}$
by going through the levels $L_i$ in order of $i = 1, 2, \ldots, q + c$, and querying \emph{all} of the edges between $v $ and the level
$L_i $; we can remove $v $ from the candidate set $C$ if we ever find an edge $ (u, v) $, where $u$ is in some $L_i$ and $u \prec v$. In particular, all vertices $u \prec x_{\ell+1}$ have already been discovered and removed from the levels $L_i$, so the existence of such a vertex $u \prec v$ implies that $v \neq x_{\ell + 1}$.

In the worst case, the process for eliminating a vertex $v$ from the candidate set $ C $ might have to look at all of the edges
incident to $v$. One of the remarkable features of our algorithm is that, because of how the levels are structured, 
the expected number of queries needed to identify $x_{\ell + 1}$ ends up being only $O(\log (pn))$. 

\paragraph{The structure of the analysis.}
In the rest of the section, we discuss how to bound the expected number of edge queries made by the algorithm.
For now, we will ignore issues surrounding whether events are independent, and we will think of events that are
 ``morally independent'' as being fully independent. One of the main challenges in the full analysis
is to define the events being analyzed in such a way that these interdependence issues can be overcome. 

Our analysis proceeds in three parts. First we analyze for each vertex $ v $ and each level $ L_i $ 
the probability that $ v\in L_i $ at any given moment. 
Next, we use this to bound the total number of queries spent maintaining the levels $L_i$.
Finally, we bound the number of queries spent eliminating candidates from candidate sets
 (i.e., actually identifying each $x_\ell$). 

\paragraph{Understanding which vertices are in which levels.}
Consider a vertex $v$ with rank $r(v) \approx s_j$. We will argue that $ v $
is likely to appear in the levels $ L_j, L_{ j +1 },\ldots$ and is unlikely to appear
in the levels $L_{j - c}, L_{j - c - 1}, \ldots, L_1$. 

Begin by considering $\Pr[v \in L_{j + i}]$ for some $i \ge 0$. If $v \not\in L_{j + i}$,
then there must be some vertex $u$ with rank $1 \le r(u) \le r(v) \approx s_j$ such that
the edge $(u, v)$ is in one of $E_{j + i}, E_{j + i + 1}, \ldots$. However, there are only
$O(s_j) = O(2^j / p)$ \emph{total} vertices $u$ with ranks $1 \le r(u) \le r(v)$, and only a $p$-fraction of them
have an edge to $ v $. Moreover, only a $O(1/ 2^{i + j})$ fraction of those edges
are expected to appear in one of $E_{j + i}, E_{j + i + 1}, \ldots$. So the probability
of any such edge appearing is that most
$$O((2^j / p) \cdot p / 2^{i + j}) = O(1 / 2^i).$$
In other words, if $v$ has rank $r(v) \approx s_j$, then 
\begin{equation}
\Pr[v \in L_{j + i}] = 1 - O(1 / 2^i)
\label{eq:appears}
\end{equation}
for each level $L_{j + i}$. 

An important consequence of \eqref{eq:appears} is that, for any given $ L_i $,
 the vertices $v$ with ranks less than $O(s_i)$ (say, less than $s_i / 16$) 
each have probability at least $\Omega(1)$ of being in $L_i$. We refer to the set of such vertices
in $ L_i $ as the \defn{anchor set} $A_i$. At any given moment, 
we have with reasonably high probability that $|A_i| \ge s_i / 128$. 


The anchor sets $ A_i $ play a critical role in ensuring that vertices with large ranks
\emph{do not} make it into low levels. Consider again a vertex $v$ with rank $r(v) \approx s_j$
and let us bound the probability that $v \in L_{j - c - i}$ for some $i \ge 0$. Assume for the moment that
$v \in L_{j - c - i + 1}$. The only way that $ v $ can be promoted to $L_{j - c - i}$ 
is if there are no edges in $E_{j - c - i}$ that connect $v$ to an element of the anchor set
$A_{j - i}$ (in particular, if there is such an edge, then $ v $ will be blocked from promotion).
However, the anchor set $A_{j - i}$ has size at least $s_{j - i} / 128 = \frac{2^{j - i}}{128 p}$,
and thus the expected number of edges (in all of $E$) from $v$ to $A_{j - i}$ is roughly $2^{j - i} / 128$.
A $1 / 2^{j - c - i}$ fraction of these edges are expected to be in $E_{j - c - i}$. Thus,
the expected number of edges that block $v$ from promotion is roughly
$$\frac{2^{j - i} / 128}{2^{j - c - i}} = 2^c / 128.$$
Since these edge-blockages are (mostly) independent, the probability that no such edges block $ v $
from promotion ends up being at most
$$e^{-2^c / 128}.$$
This analysis considers just the promotion of $v$ from level $L_{j - c - i + 1}$ to level
$ L_{j - c - i}$. Applying the same analysis repeatedly, we get that for any vertex $v$ with rank $r(v) \approx s_j$, and any $i \ge 0$,
\begin{equation}
\Pr[v \in L_{j - c - i}] \le e^{-2^c i / 128}.
\label{eq:filter0}
\end{equation}
Or, to put it another way, for each level $ L_i $, and for each vertex $ v $ satisfying $r(v) \approx s_{i + k + c}$ for some $k \ge 0$, we have that
\begin{equation}
\Pr[v \in L_{i}] \le e^{-2^c k / 128}.
\label{eq:filter}
\end{equation}

The derivation of \eqref{eq:filter} reveals an important insight in the algorithm's design.
One of the interesting features of the promotion rule is that, when deciding whether 
to promote a vertex $v$ from a given level $L_{i + 1}$ into $L_{i}$, we use the edges 
in $E_{i}$ to compare $v$ not to its companions in $L_{i + 1}$ (as might be tempting to do) but instead to the 
elements of the larger level $L_{i + c}$.
For vertices $v$ with small ranks,
this distinction has almost no effect on whether $ v $ makes it into level $L_i$ (note, in particular, 
that the proof of \eqref{eq:appears} is unaffected by the choice of $c$!); but for vertices $v$ with large ranks
(i.e., ranks at least $s_{i + c}$),
this distinction makes the filtration process between levels much more effective
(i.e., the larger that $c$ is, the stronger that \eqref{eq:filter} becomes).
It is important that $ c $ not be too large, however, since otherwise
the application of \eqref{eq:filter} to a given $L_i$ would only tell us information about
vertices $v$ with very large ranks. By setting $c$ to be a large positive constant, we get the best of all worlds.

The derivation of \eqref{eq:filter} also reveals the reason why we must perform periodic rebuilds of levels. We rely on the anchor set $A_{i + c}$ to ensure that vertices $v$ with large ranks stay out low levels $L_i$, but the anchor set $A_{i + c}$ changes dynamically. Since $A_{i + c}$ takes many different values over time, the result is that a given high-rank vertex $v$ might at some point get lucky and encounter a state of $A_{i + c}$ that allows for $v$ to get promoted to level $L_i$ (despite $v$'s large rank!). The main purpose of performing regular rebuilds is to ensure that this doesn't happen, and in particular, that the vertices in level $L_i$ were all promoted into $L_i$ based on a relatively recent version of the anchor set $A_{i + c}$.

Combined, \eqref{eq:appears} and \eqref{eq:filter} give us a clear picture of which levels we should expect each vertex
to be in. Now we turn our attention to analyzing the query complexity of the algorithm.

\paragraph{Bounding the cost of level updates.}
The next step of the analysis is to bound the number of queries spent
rebuilding levels and performing incremental updates to them.

We can deduce from \eqref{eq:filter} that, at any given moment, $\E[|L_i|] = O(s_{i + c}) = O(s_i)$. 
For our discussion here we will simply assume that $|L_i| = O(s_i)$ always holds.

Now let us consider the number of queries needed to perform a promotion test (i.e., to apply the promotion rule) on a given vertex $ v $
in a given level $ L_i $. The promotion rule requires us to examine every edge in $ E_i $ that goes from $ v $
to any vertex in $ L_{ i + c } $. The expected number of such edges is roughly
\begin{equation}
\frac{p}{2^i} |L_{i + c}| = O\left(\frac{p}{2^i} s_i\right) = O(1).
\label{eq:promotion}
\end{equation}
Thus we can treat each promotion test as taking $O(1)$ queries.

We can now easily bound the total work spent rebuilding levels from scratch.
Each level $ L_i $ is rebuilt $O(n / s_i)$ times, and each rebuild
requires us to apply the promotion rule to $O(s_i)$ different vertices.
It follows that each rebuild takes $O(s_i)$ queries, and that the
 total number of queries spent performing rebuilds on $L_i$ is $O(n)$.
Summing over the $ O (\log (pn)) $ levels results in a total of
$O(n \log (pn))$ queries. 

What about the query complexity of the incremental updates to the levels? 
Recall that whenever we discover a new $ x_\ell $, we must revisit every vertex $ v $
that was formerly blocked by $x_\ell$. For each such $v$,
we must check whether $ v $ can be promoted, and if so then 
we must repeatedly promote $ v $ until it reaches a level where it is again blocked. 

Recall that each promotion test takes $O(1)$ expected queries. To bound the total number of
promotion tests due to incremental updates, we break the promotion tests into
two categories: the \defn{failed promotion tests} (i.e., the promotion tests that
do not result in a promotion) and the \defn{successful promotion tests} (i.e.,
the promotion tests that do result in a promotion).

The number of failed promotion tests is easy to bound. It is straightforward to show that
the expected number of distinct vertices $v$ that we perform a promotion test on
(each time that a new $x_\ell$ is discovered) is $O(\log (pn))$ 
(roughly speaking, there will be one such vertex per level). Each such $ v $
can contribute at most one failed promotion test. So summing over the $x_\ell$'s,
we find that the total number of failed promotion tests is $O(n \log (pn))$.

To bound the number of successful promotion tests, on the other hand,
we can simply bound the number of promotions that occur (due to incremental updates).
Each time that some vertex $ v $ is promoted from
a level $L_{i + 1}$ to a level $ L_i $, the size of $L_i$ increases by one
(and the size of $L_{i + 1}$ is unchanged). We know that the size of each $L_i$
stays below  $O(s_i)$ at all times, however, so there can be at most $O(s_i)$ 
such promotions between every two consecutive rebuilds of the level. 
Since $L_i$ is rebuilt  $O(n / s_i)$ times, the total number of 
promotions into $L_i$ is $O(n)$.
Summing over the levels, we get $O(n \log(pn))$, as desired.

\paragraph{Bounding the cost of eliminating candidates from candidate sets.}
The final and most interesting step in the analysis is to bound the total number of queries spent
eliminating candidates from candidate sets. 
Suppose we have already discovered vertices $x_1, x_2, \ldots, x_\ell$,
and we wish to discover $x_{\ell + 1}$. How many queries does it take
to identify which element of $C$ is $x_{\ell + 1}$?

For simplicity, we will assume here that all of the candidates $v \in C$ (besides $x_{\ell + 1}$)
have ranks at least $r(v) > 1/p$ (this assumption can easily be removed with a bit of extra casework).

We can start by bounding the probability that a given $ v $ is in $ C $.
In order for $v$ to be in $C$ we need both that $v \in L_1$ and that
$(x_\ell, v)$ is an edge.
With a bit of manipulation, one can deduce from 
\eqref{eq:filter} that if a vertex $v$ has rank $r(v) \ge 1/ p$, then 
$$\Pr[v \in L_1] \le O\left(\frac{1}{(r(v) \cdot p)^{10}}\right).$$
If $v \in L_1$ (but $v \neq x_{\ell + 1}$), then $\Pr[(x_\ell, v) \text{ is an edge}] = p$.
Thus
$$\Pr[v \in C] \le O\left(\frac{p}{(r(v) \cdot p)^{10}}\right).$$

If $v \in C$, then we must look at all of the edges from $v$ 
to the levels $L_1, L_2, \ldots$ until we find an edge $(u, v)$ with $u \prec v$. 
Since each level $L_i$ has size $O(s_i)$, the expected number of edges from
$v$ to a given $L_i$ is $O(ps_i) = O(2^i)$.  If $L_J$ is the highest level that we look at
while eliminating $v$ from $C$, then the total number of queries incurred will be roughly
$O(2^J).$

So how many levels must we look at in order to eliminate $ v $ from $ C $?
Let $ u $ be the predecessor to $ v $ in the true order, and let $L_K$
be the lowest level that contains $ u $. Since $u \prec v$ and since
$(u, v)$ is deterministically an edge, we are guaranteed to stop at a level
$J \le K$. We can bound $K$ (and thus $J$) by applying \eqref{eq:appears}; 
this gives us the identity
\begin{equation}\Pr[J \ge \log (p \cdot r(v)) + i] \le O\left(1 / 2^i\right) \label{eq:Jbound}\end{equation}
for every $i \ge 0$. 

To summarize, if $v \in C$ (and $v \neq x_{\ell + 1}$), then the expected number of queries to remove $v$ from $C$ is roughly $\E[2^J]$; and for any $i \ge 0$, the probability that $2^J \ge p r(v) 2^i$ is $O(1 / 2^i)$.
There are $O(\log pn)$ possible values for $i$, each of which contributes $O(p\cdot r(v))$ to $\E[2^J]$. Thus the expected number of queries to remove $v$ from $C$ is $O(p\cdot r(v) \log (pn))$.

To bound the cost of eliminating \emph{all} candidates from $ C $, we must sum over the ranks $r \ge p^{-1}$ to get
\[{\sum_{v: r(v) \ge p^{-1}} \Pr[v \in C] \cdot \mathbb{E}[2^J]} = O\left(\sum_{r \ge p^{-1}} \frac{p}{(r p)^{9}} \log (pn)\right) = O\left(\log (pn)\right).\]
This represents the cost to identify a given $x_{\ell + 1}$. Summing over all $\ell$,
the total contribution of these costs to the query complexity of the algorithm of is
$O(n \log (pn))$.

\section{Stochastic Generalized Sorting} \label{sec:Random}

\subsection{Algorithm Design} \label{subsec:Random_Algo}


\paragraph{Conventions and Notation.}
Let $G = (V, E)$ be the input graph and let $x_1 \prec x_2 \prec \cdots \prec x_n$ denote the true order of the vertices. 

There are two types of edges in the graph, those that were included randomly 
with probability $p$, and those that are part of the true path connecting together 
$x_1, x_2, \ldots, x_n$. We say that a vertex pair $(u, v)$ is \defn{stochastic} if it is 
not an edge in the true-ordering path and otherwise the vertex pair is \defn{deterministic}. 
Note that a stochastic vertex pair does not need to have an edge connecting it, but 
a deterministic vertex pair does.

Our algorithm will find $x_1, x_2, \dots$ one at a time. Once we have found $x_i$, we say that $x_i$ has been \defn{discovered}.
As a convention, we will use $x_\ell$ to denote the most recently discovered vertex.
When we refer to the \defn{rank} of a vertex, we mean its rank out of the not yet identified vertices $\{x_{\ell + 1}, \ldots x_n\}$. For a vertex $v$, we let $r(v)$ be the rank of $v$, so $v = x_{\ell+r(v)}$.

To simplify our discussion in this section,  we will consider only the task of discovering the vertices $x_1,\ldots,x_{n/2}$.  This allows for us to assume that the number of remaining vertices is always $\Theta(n)$.  Moreover, by a symmetric argument, we can recover $x_n,x_{n-1},\ldots,x_{n/2+1}$ and therefore recover the complete order, so this assumption does not affect the correctness of the overall algorithm.

\paragraph{Constructing edge sets \boldmath $E_1, E_2, \ldots$.}
The first step of our algorithm is to use the following proposition to construct $q = O(\log (pn))$ sets of edges $E_1, E_2, \ldots, E_q$ (these sets are built once and then never modified). When discussing these sets of edges, we will use $E_i(u, v)$ to denote the indicator random variable for whether $(u, v) \in E_i$, and we will use $E(u, v)$ to denote the tuple $\langle E_1(u, v), E_2(u, v), \ldots, E_q(u,v)\rangle$. The proof of Proposition \ref{prop:Ei} is deferred to Subsection \ref{subsec:Random_Analysis}.

\begin{proposition} \label{prop:Ei}
Suppose we are given $E$ but we are not told which vertex pairs are stochastic/deterministic. Suppose, on the other hand, that Alice is told which vertex pairs are stochastic/deterministic, but is not given $E$. Based on $E$ alone, we can construct sets of edges $E_1, E_2, \ldots, E_q$ such that $E = \bigcup_{i = 1}^q E_i$, and such that, from Alice's perspective, the following properties hold: 
\begin{itemize}
\item For each stochastic vertex pair $(u, v)$, 
$\Pr[E_i(u, v)] = \frac{\alpha \cdot p}{2^{i}}$, where $1 \le \alpha \le 2$ is some constant that depends only on $p$ and $n$.
\item For each deterministic vertex pair $(u, v)$,
 $\Pr[E_i(u, v)] = \frac{\alpha}{2^i}$.
\item The random variables $\{E_i(u, v) \mid i \in [q], (u, v) \text{ is stochastic}\}$ are mutually independent.
\item The random variables $\{E(u, v)\}_{u, v \in V}$ are mutually independent.
\end{itemize}
\end{proposition}

We remark that the sets $E_i$ do not partition $E$, as the sets $E_i$ are not necessarily disjoint. The sets $E_i$ are constructed independently across edges, but for each edge $(u, v)$ in $E$ we anti-correlate the events $\{E_i(u, v)\}_{i=1}^{q}$ so that the events are independent if we do not condition on the set $E$. We formally show how to construct the sets $E_i$ in the proof of the above proposition, in Subsection \ref{subsec:Random_Analysis}.

One thing that is subtle about the above proposition is that the randomization is coming from two sources, the random generation of $E$ and the random construction of $E_1, E_2, \ldots, E_q$ based on $E$. The probabilities in the bullet points depend on both sources of randomness simultaneously, and they \emph{do not} condition on the final edge set $E$ (hence, the use of Alice in the proposition statement).

\paragraph{Assigning the vertices to levels.}
A central part of the algorithm design is to use the edge sets $E_1, E_2, \ldots, E_q$ in order to dynamically assign the vertices to \defn{levels} $L_1 \subseteq L_2 \subseteq \cdots \subseteq L_{q + c}$, where $q = O(\log (pn))$ is as defined above and $c$ is some large but fixed constant. Before we can describe how the algorithm discovers the true order of the vertices, we must first describe
how the levels $\{L_i\}$ are constructed and maintained over time. In particular, the algorithm will make use of the
structure of the levels in order to efficiently discover new vertices in the true order.

Intuitively, the goal of how the levels are maintained is to make it so that,
at any given moment, each level $ L_i $ has size roughly
$ s_i = 2 ^ i/p $, and so that the majority of the elements in $ L_i $ have ranks $O(s_i)$. We refer to
$s_i = 2^i / p$ as the \defn{target size} for level $L_i$. 

\paragraph{Initial construction of levels.}
We perform the initial construction of the levels  as follows. The final levels $L_{q + 1}, \ldots, L_{q + c}$ all automatically contain all of the vertices. For each other level $L_i$,
we construct the level as
\begin{equation}
L_{i + 1} \setminus \{v \in L_{i + 1} \mid (u, v) \in E_{i} \text{ for some } u \prec v \text{ such that }u \in L_{i + c}\}.
\label{eq:filterr}
\end{equation}
The vertices $v$ from $L_{i + 1}$ that are not included in $L_i$ are said to be \defn{blocked by the edge} $(u, v)$ defined in \eqref{eq:filterr}.

One point worth highlighting is that, when we are deciding whether a vertex $v$ should move from $L_{i + 1}$ to $L_i$,
we query not just the edges in $E_i$ that connect $v$ to other vertices in $L_{i + 1}$, but also the edges that connect $v$ to other vertices in the larger set $L_{i + c}$.
This may at first seem like a minor distinction, but as we shall later see, it is critical to ensuring that vertices with large ranks do not make to levels $L_i$ for small $i$. 

\paragraph{Performing incremental updates to levels.}
When we discover a given $x_\ell$, we perform \defn{incremental updates} to the levels as follows.
First, we remove $x_\ell$ from all of the levels. 
Next we consider vertices $v$ that were formerly blocked by an edge of the form $(x_\ell, v)$ in some level $L_i$:
since $x_\ell$ is no longer in the levels, $(x_\ell, v)$ no longer blocks $v$.
Since $x_\ell$ has been removed from the levels, we must give $v$ the opportunity to advance to lower levels.
That is, if there is no edge $(u, v) \in E_{i - 1} \cap L_{i + c - 1}$ such that $u \prec v$, then we advance $v$ to level $L_{i - 1}$; 
if, additionally, there is no edge $(u, v) \in E_{i - 2} \cap L_{i + c - 2}$, then we advance $v$ to level $L_{i - 2}$, and so on. 
The vertex continues to advance until it is either in the bottom level $L_1$ or it is again blocked by an edge.

\paragraph{Rebuilding levels periodically.}
In addition to the incremental updates, we also periodically reconstruct each level in its entirety. Recall that we use $s_i = 2^i/p$ to denote the target size for each level $L_i$. For each $i$, every time that the index $\ell$ of the most recently discovered $x_\ell$ is a multiple of $s_i / 32$, we rebuild the levels $L_i, L_{i - 1}, \ldots, L_1$, one after another, from scratch according to \eqref{eq:filterr}. That is, we rebuild level $L_i$ based on the current values of $L_{i + 1}$ and $L_{i + c}$, we then rebuild $L_{i - 1}$ based on the new value of $L_i$ and the current value of $L_{i + c - 1}$, etc. 

Note that, when rebuilding a level $L_j$ from scratch, we also redetermine from scratch which of the vertices in $L_{j + 1}$ are blocked from entering $L_j$. In particular, a vertex $v$ may have previously been blocked from entering $L_j$ by an edge $(u, v)$ for some $u \in L_{j + c}$; but if $L_{j + c}$ has since been rebuilt, then $u$ might no longer be in $L_{j + c}$, and so $v$ may no longer be blocked from entering level $L_j$.

\paragraph{Finding the first vertex.}
We are now ready to describe how the algorithm discovers the true order $x_1, x_2, \ldots$ of the vertices.
We discover the first vertex $x_1$ in a different way from how we discover the other vertices (and, in particular, we do not make use of either the edge sets $E_i$ or the levels $L_i$). 

The algorithm for finding $x_1$ works as follows.
We always keep track of a single vertex $v_0$ which will be the earliest vertex found so far (in the true ordering) and a set $S = \{v_1, \dots, v_r\}$ of vertices that we know come after $v_0$ in the true ordering.
We begin by picking an arbitrary edge $(u, v)$ and querying the edge to find which vertex precedes the other. 
If $u$ precedes $v$, we set $v_0 = u$ and $v_1 = v$; else, we set $v_0 = v$ and $v_1 = u$.
For each subsequent step, if we currently have $v_0$ and $S = \{v_1, \dots, v_r\},$ then we do the following. If there exists any edge connecting $v_0$ and some $u \not\in \{v_1, \dots, v_r\},$ we query the edge $(u, v_0).$ If $v_0$ precedes $u$, then we add $v_{r+1} := u$ to the set $S$. Otherwise, if $u$ precedes $v$, we add $v_{r+1} = v_0$ to the set $S$, and then replace $v_0$ with $u$. Finally, if there is no edge connecting $v_0$ with some $u \not\in \{v_1, \dots, v_r\},$ we obtain that $x_1 = v_0$.

\paragraph{Finding subsequent vertices.}
Now suppose we have most recently discovered some vertex $x_\ell$ and we wish to discover $x_{\ell + 1}$. This is done using the levels $\{L_i\}$ and the edge sets $\{E_i\}$ as follows.

We start by constructing a candidate set $C$ consisting of all of the vertices in $L_1$ that have an edge to $x_\ell$.
Next, we perform the following potentially laborious process to remove vertices from the candidate set until we have gotten down to just one candidate.
We go through the levels $L_1, L_2, L_3, \ldots$, and for each level $L_i$ we query \emph{all} of the edges between vertices in $L_i$ and the
remaining vertices in $C$.
We remove a vertex $v$ from the candidate set $C$ if we discover an edge $(u, v)$ with $u \prec v$. Once we have narrowed down the candidate set to a single vertex, we conclude that the vertex is $x_{\ell + 1}$.

By repeating this process over and over, we can discover all of the vertices $x_1, x_2, \ldots, x_{n / 2}$.


\subsection{Algorithm Correctness} \label{subsec:alg_correct}

In this subsection, we prove \Call{StochasticSort}{} always outputs the correct order (with probability $1$). 
First, we show that the algorithm finds $x_1$ correctly. 
\begin{proposition}
    The algorithm for finding $x_1$ always succeeds. Moreover, it deterministically uses at most $n$ queries.
\end{proposition}

\begin{proof}
    First, we note that at any step, if we have vertex $v_0$ and set $S = \{v_1, \dots, v_r\}$, then $v_0 \prec v_i$ for all $v_i \in S$. This is obviously true at the first step after we compare the first two vertices $(u, v)$, and it continues to be true later on by the following inductive argument. If we find some edge $(u, v_0)$ with $u \not\in S$, then if $u \prec v_0$, then $v \prec v_0 \prec v_i$ for all $v_i \in S$, so adding $v_0$ to $S$ and replacing $v_0$ with $u$ means the new $v_0$ still precedes all vertices in $S$. On the other hand, if $v_0 \prec u,$ then since we just add $u$ to $S$, we also still have that $v_0$ precedes everything in $S$. Thus we always have $v_0 \prec v_i$ for all $v_i \in S$.
    
    Next we show that, whenever our algorithm for finding $x_1$ terminates, it is guaranteed to have successfully found $x_1$. That is, we show that if $v_0$ has no edges $(u, v)$ such that $u \not\in S$, then $v_0$ must equal $x_1$. This is because if $v_0 \neq x_1,$ then $v_0$ is connected to its immediate predecessor in the true ordering (which we can call $v'$), which would mean $v' \in S$. This, however, contradicts the fact that $v_0 \prec v_i$ for all $v_i \in S$. Therefore, our algorithm is correct assuming it terminates.
    
    Finally, observe that the algorithm must terminate after $n$ queries, since each query increases the size of $S$ by $1$ (we either add $v_0$ or $u$ to $S$, neither of which was in $S$ before). This concludes the proof.
\end{proof}

Next, we give a condition that guarantees that a given vertex $v$ will be in a given level $L_i$.


\begin{proposition} \label{prop:correct1}
    Fix a vertex $v$ that has not been discovered yet, and suppose that
    at some point,
    for all $u \prec v$ such that $u$ has not been discovered, $(u, v) \not\in \bigcup_{j \ge i} E_j$. Then, $v$ is in level $L_i$.
\end{proposition}

\begin{proof}
    Suppose that $v \not\in L_i$. Then, since $v$ is not in the first level, the algorithm currently states that $v$ is blocked by some edge $(u, v)$ with $u \prec v$.
    
    Since the algorithm currently states that $v$ is blocked by $(u, v)$, the vertex $u$ cannot have been discovered yet. Indeed, once $u$ is discovered, we remove $u$ from all levels and perform incremental updates to all of the vertices that edges incident to $u$ blocked, including $v$. This incremental update would either push $v$ all the way to level $L_1$ or would result in $v$ being blocked by a new edge (different form the edge $(u, v$)), contradicting the fact that $v$ is currently blocked by $(u, v)$. So, $u \prec v$ and $u$ is not discovered. Finally, since $v \not\in L_i$, there exists some $j \ge i$ such that $v \in L_{j+1} \backslash L_j$, which means that since the algorithm decided that $(u, v)$ blocks $v$ from level $L_j$, we have that $(u, v) \in E_j$.
\end{proof}

As a direct corollary, we have the following:

\begin{corollary} \label{cor:correct2}
    After $x_1, \dots, x_{\ell}$ are discovered and all level updates are performed, $x_{\ell+1} \in L_1$.
\end{corollary}

\begin{proof}
    Indeed, there are no vertices $u \prec x_{\ell+1}$ such that $u$ has not been discovered. So, by setting $i = 1$ and $v = x_{\ell+1}$ in Proposition \ref{prop:correct1}, the corollary is proven.
\end{proof}

To finish the proof of algorithm correctness, suppose we have discovered $x_\ell$ and we are now in the process of discovering $x_{\ell + 1}$. Since $x_{\ell+1} \in L_1$ (by Corollary \ref{cor:correct2}), the set $C$ of vertices in $L_1$ that are connected to $x_\ell$ contains $x_{\ell+1}$. For each $v \in C$ such that $v \neq x_{\ell+1}$, the immediate predecessor of $v$ is undiscovered, so it is in $L_{q+1}$. Moreover, the immediate predecessor of $v$ has an edge connecting it to $v$. Therefore, the other vertices $v \in C$ will be eliminated eventually. However, $x_{\ell+1}$ can never be eliminated, since $x_1, \dots, x_{\ell}$ are the only vertices that precede $x_{\ell+1}$, and they have been removed from all levels. Therefore, eventually we will narrow down the candidate set to precisely $x_{\ell+1}$, meaning that we will successfully discover the correct value for $x_{\ell + 1}$. The proof of correctness follows by induction.

\subsection{Analyzing the Query Complexity} \label{subsec:Random_Analysis}

In this section we bound the total number of queries made by the \Call{StochasticSort}{} algorithm. First, as promised in Subsection \ref{subsec:Random_Algo}, we prove Proposition \ref{prop:Ei}.

\begin{proof}[Proof of Proposition \ref{prop:Ei}]
    Recall that $q$ (i.e., the number of edge sets $E_1, E_2, \ldots$) is defined solely as a function of $p$ and $n$. Choose $\alpha$ so that
\begin{equation}
    \prod_{i = 1}^{q} \left(1 - \frac{\alpha \cdot p}{2^i}\right) = 1-p.
    \label{eq:alpha}
\end{equation}
    Note that $g(\alpha) := \prod_{i = 1}^{q} \left(1 - \frac{\alpha \cdot p}{2^i}\right)$ is a continuous and strictly decreasing function over $\alpha > 0$. Moreover, $g(\alpha) \ge 1 - \sum_{i = 1}^{q} \frac{\alpha \cdot p}{2^i} \ge 1 - \alpha \cdot p$ and $g(\alpha) \le 1 - \frac{\alpha \cdot p}{2}.$ Thus, there is a unique solution $\alpha > 0$ to $g(\alpha) = 1-p$, and the solution $\alpha$ must be in the range $[1, 2]$.
    
    Let $\mathcal{D}$ be the probability distribution over tuples of the form $\langle X_1, X_2, \ldots, X_q\rangle$, where each $X_i$ is an indicator random variable that is independently set to $1$ with probability $\alpha p / 2^i$. 
    To construct the $E_i$s, for each edge $e \in E,$ we sample $E(u, v)$ at random from the distribution $\mathcal{D}$ conditioned on at least one of the $X_i$s being $1$; and for each edge $e \not\in E$, we set $E(u, v)$ to be the zero tuple.  Note that, by design, $\bigcup_i E_i = E$.
    
    Now let us analyze the sets $E_i$ from Alice's perspective (i.e., conditioning on which edges are stochastic/deterministic but not on which edges are in $E$). For each stochastic vertex pair $(u, v)$, the pair $(u, v)$ is included in $E$ with probability $p$. This is exactly the probability that a tuple sampled from $\mathcal{D}$ is non-zero. Thus, for each stochastic vertex pair $(u, v)$, we independently have $E(u, v) \sim \mathcal{D}$. (This implies the first property.) On the other hand, for each deterministic vertex pair $(u, v)$, we independently have $E(u, v) \sim \mathcal{D} \mid \left(\OR_i X_i = 1\right)$.
    Observe that for $\langle X_1, X_2, \ldots, X_q \rangle \sim D$,
    \begin{align*}
    \Pr\left[X_i \mid \OR_i X_i = 1\right] & = \frac{\Pr[X_i]}{\Pr[\OR_i X_i = 1]} \\
    & = \frac{\Pr[X_i]}{p}  \\
    & = \frac{\alpha}{2^i},\end{align*}
    where the first equality uses the fact that $X_i$ can only hold if $\OR_i X_i$ holds, and the second equality uses \eqref{eq:alpha}.
    This establishes that for any deterministic vertex pair $(u, v)$, $\Pr[E_i(u, v)] =  \frac{\alpha}{2^i}$, hence the second property. Finally, the aforementioned independencies imply the third and fourth properties, with the third property (independence if we do not condition on $E$) also using the fact that we chose $\alpha$ to satisfy Equation \eqref{eq:alpha}.
\end{proof}

Throughout the rest of the algorithm analysis, whenever we discuss the events $\{E_i(u, v)\}$, we will be taking the perspective of Alice from Proposition \ref{prop:Ei}. That is, our analysis will be considering both the randomness that is involved in selecting stochastic edges, and the randomness that our algorithm introduces when constructing the $E_i$s.

For now, we shall assume that $1/p$ is at least a sufficiently large constant multiple of $\log n$; this assumption means that the target size $s_i$ for each level $L_i$ satisfies $s_i = \Omega(\log n)$ for all $i$, and will allow for us to use high probability Chernoff bounds in several places.
We will remove the assumption at the end of the section and extend the analysis to consider arbitrary values of $p$.

Assume that the first $\ell$ vertices $x_1, \dots, x_{\ell}$ have been discovered, and we are in the process of identifying $x_{\ell+1}$.
Recall that, for each vertex pair $(u, v)$, and for each level $L_i$, we use $E_i(u, v)$ to denote the indicator random variable for the event that $(u, v) \in E_i$. 
Let $X_{i, v}$ be the indicator random variable
$$\OR_{j \ge i} \OR_{x_\ell \prec u \prec v} E_j(u, v)$$
for the event that any vertex pair $(u, v)$ with $x_\ell \prec u \prec v$ is contained in any of $L_i, L_{i + 1}, \ldots$. Note that there is no restriction on $u$ as to which levels contain it (all $x_\ell \prec u \prec v$ are considered). For a given vertex $v$ satisfying $x_\ell \prec v$, if $X_{i, v} = 0$ then $v$ is guaranteed to be in $L_i$, but the reverse does not hold ($v$ could be in $L_{i}$ despite $X_{i, v}$ being $1$).

The random variables $X_{i, v}$ are independent across vertices (but not across levels), since even for deterministic edge pairs we still have independence of $E(u, v) = \{E_{i}(u, v)\}_{i = 1}^{q}$ across vertex pairs (see the fourth property of Proposition \ref{prop:Ei}). Moreover, each $X_{i, v}$ satisfies the following useful inequality:
\begin{lemma} \label{lem:X}
Suppose $v$ has rank $r$. Then
$$\Pr[X_{i, v}] \le \frac{\alpha \cdot (1 + rp)}{2^{i - 1}},$$
where $\alpha$ is the positive constant defined in the construction of the $E_i$s.
\end{lemma}
\begin{proof}
By the union bound,
\begin{align*}
\Pr[X_{i, v}] & \le \sum_{j \ge i} \phantom{f} \sum_{x_{\ell} \prec u \prec v} \Pr[E_j(u, v)]. \\
\end{align*}
Recall from Proposition \ref{prop:Ei} that each stochastic vertex pair has probability at most $\alpha \cdot p / 2^j$ of being included in $E_j$,
and that each deterministic vertex pair has probability at most $\alpha / 2^j$ of being included in $E_j$.
Thus 
\begin{align*}
\Pr[X_{i, v}] & \le \sum_{j \ge i} \left(\frac{\alpha}{2^j} + (r - 1) \frac{\alpha \cdot p}{2^j}\right) \\
              & \le \frac{\alpha \cdot (1 + rp)}{2^{i - 1}}. \qedhere
\end{align*}
\end{proof}

Whenever a level $L_i$ is rebuilt, we define the \defn{anchor vertices} $A_i$ to be the vertices in $L_i$ with ranks in the range 
\begin{equation*}
    [s_i / 16, 3s_i / 32). 
\end{equation*}
Note that, even as $L_i$ changes incrementally over time, the anchor set $A_i$ does not change until the next time that $L_i$ is rebuilt from scratch.
Moreover, because $L_i$ is rebuilt relatively frequently (once for every $s_i / 32$ vertices $x_\ell$ that are discovered),
no anchor vertices are ever removed from $L_i$ (until after the level is next rebuilt). 
Moreover, the rank of any anchor vertex $v \in A_i$ is always between $s_i/32$ and $3s_i/32$, since the rank is initially at least $s_i / 16$ and decreases by a total of $s_i / 32$ between rebuilds.

Our next lemma establishes that, with high probability, there are a reasonably large number of anchor vertices in each level at any given moment.
Recall that $r(v)$ is used to denote the rank of a given vertex $v$.

\begin{lemma} \label{lem:anchor}
Let $i > 3$. Then, with probability at least $1-n^{-10}$,
$$|A_i| \ge s_i / 128.$$
\end{lemma}
\begin{proof}
For any vertex $v$, if $X_{i, v} = 0$, then $v$ must be in $L_i$, as there is nothing that can block it from levels $i$ or above. Therefore, 
$$|A_i| \ge \sum_{v \;: \; r(v) \in [s_i/16, 3s_i/32)} (1 - X_{i, v}).$$

By Lemma \ref{lem:X},
\begin{align*}
\E\left[\sum_{v \; : \; r(v) \in [s_i/16, 3s_i/32)} (1-X_{i, v}) \right] & \ge \frac{s_i}{32} \left(1 - \frac{\alpha\left(1 + \frac{3s_i}{32}p\right)}{2^{i - 1}}\right) \\
          &  \ge  \frac{s_i}{32} \left(1 - \frac{2\left(1 + \frac{3}{32} \cdot 2^i\right)}{2^{i - 1}}\right) \\
          & \ge \frac{3s_i}{256}, 
\end{align*}
since $\alpha \le 2$ and $i \ge 4.$ Moreover, since the variables $X_{i, v}$ are independent across vertices $v$ (by the fourth property of Proposition \ref{prop:Ei}), we can apply a Chernoff bound to deduce that $|A_i| \ge s_i / 128$ with probability at least $1-n^{-10}$, as we are assuming $s_i$ is at least a sufficiently large constant multiple of $\log n$.
\end{proof}

So far, dependencies between random variables have not posed much of an issue. In the subsequent lemmas, however, 
we will need to be careful about the interaction between which vertices are in each level, which vertices are in each anchor set, and which random variables $E_j(u_1, u_2)$ hold. 
For this, we will need the following two propositions. 

\begin{proposition} \label{prop:Ai_dependencies}
    At any time in the algorithm, for each vertex $v$ and level $i$, the event that $v \in L_i$ only depends on $E_j(u_1, u_2)$ over triples $(j, u_1, u_2)$ where $j \ge i$ and $u_1, u_2 \preceq v$ (including if $u_1$ or $u_2$ is already discovered).
\end{proposition}

\begin{proof}
    We prove this by induction on $i$. For $i > q,$ the proposition is trivial since every vertex $v$ (that is not yet discovered) is in level $L_i$. Assume the claim is true for levels $i+1, i+2, \dots$. 
    
    Define the sets $L'_j = L_j \cap \{u \mid u \preceq v\}$. By the inductive hypothesis, the sets $L'_{i + 1}, L'_{i + 2}, \ldots$ depend only on $E_j(u_1, u_2)$ over triples $(j, u_1, u_2)$ where $j \ge i + 1$ and $u_1, u_2 \preceq v$. If we fix the outcomes of those $E_j(u_1, u_2)$s, thereby fixing the outcomes of $L'_{i + 1}, L'_{i + 2}, \ldots$, then whether or not $v \in L_i$ depends only on $E_j(u, v)$ where $i = j$ and $u \preceq v$. Thus whether or not $v \in L_i$ depends only on the allowed $E_j(u_1, u_2)$ variables.
\end{proof}


\begin{proposition} \label{prop:independence_2}
    Conditioned on the sets $A_i$ over all $i$, the random variables $E_i(u, v)$ are jointly independent over all triples $(i, u, v)$ with $u \in A_{i+c}$ and with $v$ satisfying $r(v) \ge s_{i+c}/8$. Moreover, for each such triple $(i, u, v)$, the probability $\Pr[E_i(u, v) = 1]$ remains $\alpha \cdot p/2^i$, even after conditioning on the sets $A_i$.
\end{proposition}
\begin{proof}
    Note that $A_i$ only depends on which $v$ of ranks between $s_i/16$ and $3s_i/32$ are in level $L_i$. Therefore, by Proposition \ref{prop:Ai_dependencies}, $\{A_i\}$ over all $i$ is strictly a function of $E_{j}(u, v)$ over all choices $(i, j, u, v)$ with $j \ge i$ and $u, v \preceq x_{\ell+(3s_i/32)} \preceq x_{\ell + (3s_j/32)}$ (possibly including already discovered vertices $u, v$).  Simplifying, we get that $\{A_i\}$ over all $i$ is strictly a function of $E_{j}(u, v)$ over all choices of $(j, u, v)$ satisfying $u, v \preceq x_{\ell + (3s_j/32)}$. Let $T_1$ denote the set of such triples $(j, u, v)$.

    
    We wish to prove the independence of $\{E_i(u, v)\}$ over the set $T_2$ of triples $(i, u, v)$ with $u \in A_{i+c}$ and $r(v) \ge s_{i+c}/8$. 
    Note that if a triple $(i, u, v)$ is in $T_2$, then we must have that $x_{\ell + (3s_i/32)} \prec u, v$, so $(i, u, v) \not\in X$. Moreover, $r(v) \ge s_{i+c}/8 \ge 3s_{i+c}/32 + 2 \ge r(u)+2,$ so $(u, v)$ is a stochastic vertex pair. Thus, the triples $(i, u, v) \in T_2$ do not include any deterministic vertex pairs $(u, v)$, and are disjoint from the triples $T_1$ on which the anchor sets $A_i$ depend. The conclusion follows from the first and third properties of from Proposition \ref{prop:Ei}.
\end{proof}

Now, for each vertex $v$ satisfying $x_\ell\prec v$ and for each level $L_i$ such that $r(v) \ge s_{i + c} / 8$,
define the indicator random variable
$$Y_{i, v} = \OR_{u \in A_{i + c}} E_i(u, v).$$
In order for $v$ to be in any of the levels $L_1, \ldots, L_i$, we must have that $Y_{i, v} = 0$.

The next lemma bounds the probability that $Y_{i, v} = 0$ for a given vertex $v$.

\begin{lemma}
Consider a vertex $v$ satisfying $x_\ell\prec v$, and suppose that we condition on the $A_i$s such that each $A_i$ has size at least $s_i/128$. Then, for each level $L_i$ we have
$$\Pr[Y_{i, v}] \ge 1 - e^{-2^c / 128}.$$
Moreover, conditioned on the $A_i$s, the $Y_{i,v}$s are mutually independent.
\label{lem:Y}
\end{lemma}
\begin{proof}
By Proposition \ref{prop:independence_2}, we can use the independence of the $E_i(u, v)$s conditioned on the $A_i$s to conclude that
\begin{align*}
\Pr[Y_{i, v}] & \ge 1 - \left(1 - \frac{\alpha \cdot p}{2^i}\right)^{|A_{i + c}|} \\
         & \ge 1 - \left(1 - \frac{p}{2^i}\right)^{s_{i + c} / 128} \\
         & = 1 - \left(1 - \frac{p}{2^i}\right)^{2^{i + c} / (128 p)} \\
         & \ge 1 - e^{-2^c / 128}.
\end{align*}
Moreover, conditioned on the $A_i$s, Proposition \ref{prop:independence_2} implies that the random variables $Y_{i, v}$ are all mutually independent.
\end{proof}

Using these observations, we can now bound to the size of each $L_i$.
\begin{lemma}\label{lem:size}
Each $L_i$ has size at most $O(s_i)$ with probability at least $1 - 2n^{-10}$.
\end{lemma}
\begin{proof}
Recall from Lemma \ref{lem:anchor} that, with probability at least $1 - n^{-10}$,  we have $|A_j| \ge s_j / 128$ for every $A_j$. Condition on some fixed choice of the $A_j$s, such that $|A_j| \ge s_j / 128$ for each $A_j$. Fix $i$, and for each vertex $v$ with $r(v) \ge s_{i+c}/8$, let 
$$Z_v = \AND_{j > i \text{ such that } s_{j + c} / 8 \le r(v)} \phantom{f} (Y_{j, v} = 0).$$
We claim that, if $v \in L_i$, then $Z_v$ must occur.
If $Z_v$ does not occur, then there is some $j > i$ such that $s_{j+c}/8 \le r(v)$ and some $u \in A_{j+c}$ such that $E_j(u, v) = 1.$ But then, at the last time level $L_{j+c}$ (and thus all lower levels) was rebuilt, $u$ would block $v$ from coming to level $j$ (and thus from coming to level $i$), and since $u$ has not been removed yet, $v$ must not be in $L_i$.

For a given $v$ with $r(v) \ge s_{i+c}/8$, since $Z_v$ depends on $\log (r(v) / s_i) - O(1)$ different $Y_{j, v}$s, we have by Lemma \ref{lem:Y} that
$$\Pr[Z_v] \le (e^{-2^c / 128})^{\log (r(v) / s_i) - O(1)} \le 4^{-\log (r(v) / s_i)} \le (s_i/r(v))^2,$$
assuming the constant $c$ is sufficiently large. The expected number of $v$ satisfying $r(v) \ge s_{i+c}/8$ for which $Z_v$ holds is therefore at most 
$$\sum_{r = s_{i+c}}^n (s_i / r(v))^2 \le O(s_i).$$
By a Chernoff bound (which can be used since the $Z_v$'s are independent by Lemma \ref{lem:Y}), the total number of vertices $v$ for which $r(v) \ge s_{i+c}/8$ and $Z_v$ holds is $O(s_i)$ with failure probability at most $n^{-10}$.
This, in turn, means that
$$|\{v\mid r(v) \ge s_{i+c}/8\}\cap L_i| \le O(s_i).$$

On the other hand, $L_i$ can contain at most $s_{i+c}/8$ vertices with $r(v) < s_{i+c}/8$.
Thus, $|L_i| \le O(s_i)$.
\end{proof}

Having established the basic properties of the levels, we now bound the total number of comparisons of the various components of the algorithm.

\begin{lemma}
The total expected number of comparisons spent rebuilding levels (from scratch) is $O(n \log (pn))$.
\end{lemma}
\begin{proof}
It suffices to show that, for each level $L_i$, the expected number of comparisons spent performing rebuilds on $L_i$ is $O(n)$.
We deterministically perform $O(n / s_i)$ rebuilds on $L_i$. Each time that we perform a rebuild,
we must query all of the edges in $E_i$ that go from $L_{i + 1}$ to $L_{i + c}$. 
By Lemma \ref{lem:size}, we know that $|L_{i + 1}|$ and $|L_{i + c}|$ are $O(s_i)$ with high probability in $n$.
Moreover, by Proposition \ref{prop:Ai_dependencies}, the levels $L_{i + 1}$ and $L_{i + c}$ only depend on $E_{i+1}, \dots, E_q$, so by Proposition \ref{prop:Ei} they are independent of the random variables $E_i(u, v)$ 
ranging over the stochastic vertex pairs $(u, v)$.
Therefore, if we assume that $|L_i|$ and $|L_{i + c}|$ are $O(s_i)$, then the expected number of edges that we must query for the rebuild is
$$O\left(s_i + s_i^2 \cdot \frac{\alpha \cdot p}{2^i}\right) = O(s_i),$$
since there are at most $|L_{i+1}|+|L_{i+c}| = O(s_i)$ deterministic vertex pairs and $O(s_i^2)$ stochastic vertex pairs that we might have to query, and each of the stochastic pairs is included in $E_i$ with probability $\alpha \cdot p / 2^i$.

In summary, there are $O(n / s_i)$ rebuilds of level $L_i$, each of which requires $O(s_i)$ comparisons in expectation. The total number of comparisons from performing rebuilds on $L_i$ is therefore $O(n)$ in expectation.
\end{proof}

Recall that, each time that we discover a new vertex $x_\ell$ in the true order, we must revisit
each vertex $v$ and each $L_i$ such that $v \in L_i$ and $(x_\ell, v) \in E_{i - 1}$. Because the vertex
$x_\ell$ has been discovered, it is now removed from all of the levels, and thus
the edge $(x_\ell, v)$ no longer blocks $v$ from advancing to level $L_{i - 1}$.
If there is another edge $(u, v) \in E_{i - 1}$ such that $u \prec v$ and $u \in L_{i + c - 1}$,
then the vertex $v$ remains blocked from advancing to level $L_{i - 1}$. 
If, however, $v$ is no longer blocked from advancing, and $v$ will advance some number $k \ge 1$
of levels. In this case, we say that $v$ performs $k$ \defn{incremental advancements}.

\begin{lemma}
The total expected number of queries performed for incremental advancements is $O(n \log (pn))$.
\end{lemma}
\begin{proof}
For each level $i$, we only attempt to increment vertices $v \in L_i \backslash L_{i-1}$ if $v$ is blocked by $x_\ell$, meaning that $(x_\ell, v) \in E_{i-1}$. Except for when $v = x_{\ell+1},$ the event that $(x_\ell, v) \in E_{i-1},$ which occurs with probability $\alpha \cdot p/2^{i-1} = O(p \cdot 2^{-i})$, is independent of whether $v \in L_i$, which only depends on $E_i, E_{i+1}, \ldots$ by Proposition \ref{prop:Ai_dependencies}. By Lemma \ref{lem:size}, $|L_i| = O(s_i) = O(2^i/p)$ with probability $1-O(n^{-10})$. It follows that the expected number of vertices $v \in L_i \backslash L_{i-1}$ that we even attempt to increment, in expectation, is $O(p \cdot 2^{-i} \cdot 2^i/p) = O(1)$. Over all levels, we increment an expected $O(\log (pn))$ vertices during each edge discovery (i.e., when trying to find $x_{\ell+1}$).

Whenever we try to incrementally advance a vertex $v$ that is currently at some level $L_{j+1},$ the number of 
queries that we must perform in order to determine whether $v$ should further advance to level $L_{j}$ is equal to the number of vertices $w$ in $L_{j+c}$ such that $(v, w) \in E_{j}$. However, the only information about $v$ that we are conditioning on is that $v$ is in $L_{j+1}$ and that $v$ was blocked by $x_{\ell}$. So, by Propositions \ref{prop:Ei} and \ref{prop:Ai_dependencies}, the events $E_j(v, w)$ ranging over $w$ for which $(v, w)$ is a stochastic edge are independent of the information about $v$ that we have conditioned on; note that the only $w$ for which $(v, w)$ is not stochastic are the vertices $w$ that come immediately before or immediately after $v$ in the true order (which we call $\text{prec}(v)$, $\text{succ}(v)$). Therefore, the expected number of queries we must perform is at most $2 + |L_{j + c}| \cdot \frac{\alpha \cdot p}{2^j}$, since each vertex in $L_{j + c} \setminus \{v, \text{prec}(v), \text{succ}(v)\}$ has probability $\frac{\alpha \cdot p}{2^i}$ of having an edge in $E_i$ to $v$.
 By Lemma \ref{lem:size}, we know that $|L_{i+c}| = O(s_i) = O(2^i / p)$ with probability at least $1 - 2n^{-10}$.
Thus, the expected number of queries that we must perform for each incremental advancement is $O(1)$.

To prove the lemma, it suffices to show that the expected total number of attempted incremental advancements is $O(n \log (pn))$, since each one uses an expected $O(1)$ queries.
Recall that, each time that we try to discover a new vertex $x_{\ell+1},$ we attempt to increment $O(\log (pn))$ vertices in expectation: each vertex $v$ can fail to be incremented only once, because we stop incrementing $v$ after that. So, the expected number of failed incremental advancements across the entire algorithm is $O(n \log (pn))$.
In addition, whenever we perform a successful incremental advancement, we increase the size of some level by $1$.
We know that, with high probability in $n$, each level $L_i$ never has size exceeding $O(s_i)$.
Between rebuilds of $L_i$, $O(s_i)$ total vertices are removed from $L_i$. In order so that $|L_i| = O(s_i)$,
the total number of vertices that are added to $L_i$ between rebuilds must be at most $O(s_i)$.
Thus, between rebuilds, at most $O(s_i)$ vertices are incrementally advanced into $L_i$. 
Since $L_i$ is rebuilt $O(n / s_i)$ times, the total number of incremental advancements into $L_i$ over all time is $O(n)$.
Summing over the levels $L_i$, the total number of incremental advancements is $O(n \log (pn))$.
\end{proof}

After doing the incremental advancements for a given $x_\ell$, and performing any necessary rebuilds of levels, the algorithm searches for the next vertex $x_{\ell+1}$. Our final task is to bound the expected number of queries needed to identify $x_{\ell + 1}$.

\begin{lemma}
The expected number of comparisons needed to identify a given $x_{\ell+1}$ (i.e., to eliminate incorrect candidates from the candidate set) is $O(\log (pn))$.
\end{lemma} 
\begin{proof}
Consider the candidate set $C$ for $x_{\ell+1}$.
In addition, consider a vertex $v \succ x_{\ell+1}$. For $v$ to be in the candidate set $C$, the following two events must occur:
\begin{itemize}
\item \textbf{Event 1:} $v$ must have an edge to $x_{\ell}$. 
\item \textbf{Event 2:} For all $i$ such that $r(v) \ge s_{i+c}/8$, we have $Y_{i, v} = 0$.
\end{itemize}

Event $1$ occurs with probability $p$, and only depends on $\{E_j(x_\ell, v)\}$ over all $j$. On the other hand, by Proposition \ref{prop:Ai_dependencies}, the sets $\{A_k\}_{k \ge 1}$ only depend on $\{E_j(u_1, u_2)\}$ over triples $(j, u_1, u_2)$ with $u_1, u_2 \preceq x_{\ell + 3s_j/32}$. Let 
$$T_1 = \{(j, x_\ell, v) \mid j \in [q]\} \cup \{(j, u_1, u_2) \mid u_1, u_2 \preceq x_{\ell + 3s_j/32}, j \in [q]\}$$
be the set of triples $(j, u_1, u_2)$ whose corresponding variables $E_j(u_1, u_2)$ cumulatively determine Event 1 and the anchor sets $\{A_k\}_{k \ge 1}$. 

Event 2 considers $Y_{i, v}$ for $i$ such that $s_{i+c}/8 \le r(v)$ (note that these are the only $i$ for which $Y_{i, v}$ is defined), so assuming the $A_k$s are fixed, it only depends on $E_i(u, v)$ for triples in the set 
$$T_2 = \{(i, u_1, v) \mid s_{i + c} / 8 \le r(v), u_1 \in A_{i + c}\}.$$
Note that $T_2$ is disjoint from $T_1$, since $u_1\in A_{i+c}$ means $x_\ell \prec u$ and since the fact that $s_{i + c} / 8 \le r(v)$ means that $x_{\ell + s_{i + c}/8} \preceq v$, so $(i, u_1, v) \not\in T_1$. Moreover, $T_2$ consists exclusively of stochastic vertex pairs, since for any $u_1 \in A_{i + c}$, we have $r(v) \ge s_{i+c}/8 \ge r(u_1)+2$. Thus, if we fix the outcomes of Event 1 and of the anchor sets $\{A_k\}_{k \ge 1}$, then Proposition \ref{prop:Ei} tells us that the outcomes of the $E_i(u_1, v)$s corresponding to $T_2$ are independent of the random variables that have already been fixed. 

Now let us consider the probability of Event 2 if we condition on Event 1 and if we condition on the sets $A_k$ each having size at least $s_k/128$. There are $\log (r(v) \cdot p) - O(1)$ random variables $Y_{i, v}$ for which $r(v) \ge s_{i + c} / 8$, and by Lemma \ref{lem:Y} each $Y_{i, v}$ independently has probability at most $e^{-2^c / 128}$ of being $0$ (based on the outcomes of the random variables $E_i(u_1, v)$ where $u_1$ is selected so that $(i, u_1, v) \in T_2$). So conditioned on the outcome of Event 1 and on the sets $A_k$ each having size at least $s_k/128$,
Event 2 occurs with probability at most
$$\left(e^{-2^c / 128}\right)^{\log (r(v)p) - O(1)} \le \frac{O(1)}{(r(v) \cdot p)^{10}}.$$

Let $\mathcal{E}_1$ and $\mathcal{E}_2$ be the indicator random variables for Events 1 and 2, respectively, and let $\mathcal{D}$ be the indicator that $|A_k| \ge s_k/128$ for all $k$. To summarize, we have so far shown that $\Pr[\mathcal{E}_1] = p$ and that 
$$\Pr[\mathcal{E}_2 \mid \mathcal{E}_1, \mathcal{D}] \le \frac{O(1)}{(r(v) \cdot p)^{10}}.$$


Let $u$ be the true predecessor of $v$. 
As soon as we query the edge $(u, v)$, we will be able to eliminate $v$ from the candidate set.
Thus, if $v$ is in the candidate set, then we can upper bound the number of queries needed to eliminate $v$ by
the number of query-able edges from $\{x_\ell, v\}$ to each level $L_{i + 1}$ for which $u \not\in L_{i}$.

Now consider a fixed level $L_i$. In order so that $u \not\in L_i$, we must have that:
\begin{itemize}
\item \textbf{Event 3: }$X_{i, u }= 1$.
\end{itemize}
Event 3 depends only on $E_j(u_1, u)$ where $u_1 \prec u$ and $j \ge i$. This means that the events that determine Event 3 concern different pairs of vertices than do the events $\{E_j(x_\ell,v)\}_{j = 1}^q$ that determine Event 1 (since the latter events all involve $u$'s successor $v$); and the events that determine Event 3 concern different pairs of vertices than do the events $\{E_j(u_1, v)\}_{(j, u_1, v) \in T_2}$ that determine Event 2 once the $A_k$s are fixed (since, once again, the latter events all involve $u$'s successor $v$). Importantly, by the fourth property of Proposition \ref{prop:Ei}, this means that Event $1$ and Event $3$ are independent if we do not condition on the $A_k$s, and that conditioning on Events $1$ and $3$ along with the $A_k$s does not affect the probability of Event $2$ in comparison to conditioning only on Event 1 and the $A_k$s (and not on Event 3). 

Let $\mathcal{E}_3$ be the indicator for Event 3. By Lemma \ref{lem:X},
$$\Pr[\mathcal{E}_3] \le \frac{1 + r(u)p}{2^{i - 1}}.$$
Thus we have that
\begin{align*}
    \Pr[\mathcal{E}_1, \mathcal{E}_2, \mathcal{E}_3] &= \Pr[\mathcal{E}_1, \mathcal{E}_3, \mathcal{D}] \cdot \Pr[\mathcal{E}_2 \mid \mathcal{E}_1, \mathcal{E}_3, \mathcal{D}] + \Pr[\mathcal{E}_1, \mathcal{E}_2, \mathcal{E}_3, \lnot \mathcal{D}] \\
    & \le \Pr[\mathcal{E}_1, \mathcal{E}_3] \cdot \Pr[\mathcal{E}_2 \mid \mathcal{E}_1, \mathcal{E}_3, \mathcal{D}] + \Pr[\lnot \mathcal{D}] \\
    &= \Pr[\mathcal{E}_1] \cdot \Pr[\mathcal{E}_3] \cdot \Pr[\mathcal{E}_2 \mid \mathcal{E}_1, \mathcal{D}] + \Pr[\lnot \mathcal{D}] \\
    &\le p  \cdot \frac{1+r(u) \cdot p}{2^{i-1}} \cdot \min\left(1, \frac{O(1)}{(r(v) \cdot p)^{10}}\right) + \frac{1}{n^9}.
\end{align*}
We remark that, although Events 1, 2, 3 are formally defined only for $i > 0$ (since there is no level $L_0$), if we also consider a level $L_0$ to be the empty set (so $u \not\in L_0$ by default), we have that $\Pr[\mathcal{E}_3] \le 1 \le \frac{1+r(u) p}{2^{i-1}}$, so our bound for $\Pr[\mathcal{E}_1, \mathcal{E}_2, \mathcal{E}_3]$ is still true even in the case of $i = 0$.

For any vertex $v \succ x_{\ell+1}$, if $v$'s predecessor $u$ is in $L_{i+1} \backslash L_{i}$ for some $i \ge 0$ (where $L_0 = \emptyset$), then once we have queried edges from $v$ to $L_{i+1}$ we will have found $u \neq x_{\ell + 1}$, eliminating $v$ from the candidate set. Hence, for any level $L_{i+1}$, we will only query edges from $v$ to $L_{i+1}$ if Event 3 occurs for level $i$, and we only query edges from $x_{\ell+1}$ to $L_{i+1}$ if Event $3$ occurs for some $v$ in the original candidate set (since otherwise we will have eliminated all vertices except $x_{\ell+1}$).

We have already bounded the probability of any of Events 1, 2, 3 occurring for a given $v$ and level $L_i$. Therefore, if we wish to bound the number of edges that are queried while discovering $x_{\ell + 1}$, then our final task is the following: for each $v$ and level $L_i$ (including $i = 0$),
we must bound the number $h$ 
of query-able edges $e$ from $\{x_{\ell+1}, v\}$ to $L_{i + 1}$, conditioning on Events 1, 2, 3 occurring for $v$ and $L_i$.
We do not need to count any edge $e$ that has already been queried in the past.
There are at most $O(1)$ deterministic edges incident to $\{x_{\ell+1}, v\}$. 
On the other hand, if we condition on which edges have been queried in the past, and we also condition on Events 1, 2, and 3, then each remaining stochastic vertex pair (that is not yet been queried) has conditional probability at most $p$ of being a query-able edge. 
In particular, for each stochastic vertex pair that has not been queried, the only information that our conditions can reveal about it is that there is some set of $E_i$s in which the vertex pair is known not to appear\footnote{Note, in particular, that every time that the algorithm checks whether some edge $(u, v)$ is in $E$ or is in some $E_i$, if the algorithm finds that the edge is, then it immediately queries the edge (unless $u = x_{\ell}$, in which case we already know the direction of the edge $(u, v)$ since $v \succ x_{\ell}$, so we will never need to query it). Thus, the only information that the algorithm ever learns about not-yet-queried/not-yet-solved edges is that those edges are \emph{not} in some subset of the $E_i$s. (This can also easily be verified by the pseudocode in Appendix \ref{sec:Pseudocode}.)}; and this can only decrease the probability that the vertex pair is a query-able edge.
Since $|L_{i + 1}| \le O(s_i)$ with high probability in $n$,
the expected number of edges that we must query from $\{x_{\ell+1}, v\}$ to $L_{i + 1}$ 
(and if we condition on Events 1, 2, 3) is $O(p s_i)$ (unless the combined probability 
of Events 1, 2, 3 is already $ \le 1 / \poly(n)$, in which case we can use the trivial bound of $O(n^2)$ on the number of edges that we query).\footnote{Note that $ps_i \ge \Omega(1)$, and thus the bound of $O(ps_i)$ also counts the $O(1)$ deterministic edges that we might have to query.}

Combining the pieces of the analysis, and setting $r = r(v)$, the expected number of comparisons 
that we incur removing vertices from $x_\ell$'s candidate set is at most 
\begin{align*}
&\hspace{0.5cm} \frac{1}{\poly n} + \sum_v \sum_{i = 1}^{\log (pn)} \Pr[\text{Events 1, 2, 3 for }v \text{ and }L_i] \cdot O(p s_i) \\
& = \frac{1}{\poly n} + \sum_{r = 1}^n \sum_{i = 1}^{\log (pn)} O\left(p \cdot \min\left(1, \frac{1}{(r \cdot p)^{10}}\right) \cdot \frac{1 + rp}{2^{i  - 1}} \cdot p s_i\right)  \\
& = \frac{1}{\poly n} + \sum_{r = 1}^n \sum_{i = 1}^{\log (pn)} O\left(p \cdot \min\left(1, \frac{1}{(r \cdot p)^{10}}\right) \cdot \frac{1 + rp}{2^{i  - 1}} \cdot 2^i\right) \\
& = \frac{1}{\poly n} + \sum_{r = 1}^n \sum_{i = 1}^{\log (pn)} O\left(p \cdot \min\left(1, \frac{1}{(r \cdot p)^{10}}\right) \cdot (1 + rp) \right) \\
& = \frac{1}{\poly n} + O(\log (pn)) \sum_{r = 1}^n O\left(p \cdot \min\left(1, \frac{1}{(r \cdot p)^{10}}\right) \cdot (1 + rp) \right).\\
\end{align*}
    Furthermore, we can write the sum 
\begin{align*}
    \sum_{r = 1}^n \left(p \cdot \min\left(1, \frac{1}{(r \cdot p)^{10}}\right) \cdot (1 + rp) \right) &= \sum_{r = 1}^{1/p} p \cdot (1 + rp) + \sum_{r = 1/p+1}^{n} \frac{p \cdot (1+rp)}{(r \cdot p)^{10}} \\
    &\le \frac{1}{p} \cdot p \cdot 2 + 2p \cdot \int_{r = 1/p}^{\infty} \frac{1}{(r \cdot p)^{9}} dr \\
    &\le 2 + 2 \cdot \int_{r' = 1}^{\infty} \frac{1}{(r')^9} \cdot dr' = O(1),
\end{align*}
    substituting $r' = r \cdot p$. Therefore, in total, the expected number of comparisons  that we incur while removing vertices from $x_\ell$'s candidate set is $O(\log (pn)).$
\end{proof}


The preceding lemmas, along with subsection \ref{subsec:alg_correct}, combine to give us the following theorem:
\begin{theorem}
The expected number of comparisons made by the \Call{StochasticSort}{} algorithm is $O(n \log (pn))$. Moreover, the algorithm always returns the correct ordering.
\label{thm:main}
\end{theorem}

We conclude the section by revisiting the requirement that $1/p$ is at least a sufficiently large constant multiple of $\log n$.
Recall that this requirement was so that all of the level sizes would be large enough that we could apply Chernoff bounds to them.
We now remove this requirement, thereby completing the proof that Theorem \ref{thm:main} holds for all $p$. 
Consider an index $i$ for which $s_i$ is a large constant multiple of $\log n$. We can analyze all of the levels $L_j$, $j \ge i$,
as before, but we must analyze the levels below level $i$ differently.
One way to upper bound the number of edge queries
in the levels below $L_i$ is to simply assume that every pair of vertices in $L_i$ is queried as an edge. 
Even in this worst case, this would only add $O(\log^2 n)$ queries between every pair of consecutive rebuilds of $L_i$
(since the number of distinct vertices that reside in $L_i$ at any point between a given pair of rebuilds is $O(\log n)$).
There are $O(n / \log n)$ total rebuilds of $L_i$, and thus the total number of distinct queries that occur at lower levels
over the entire course of the algorithm is upper bound by $O(n \log n)$. Thus Theorem \ref{thm:main} holds even when $1/p = O(\log n)$.

\section{Lower Bound for Stochastic Generalized Sorting} \label{sec:Lowerbound}

In this section, we prove that our stochastic generalized sorting bound of $O(n\log(pn))$ is tight for $p \ge \frac{\ln n + \ln \ln n + \omega(1)}{n}$. In other words, for $p$ even slightly greater than $(\ln n)/n$, generalized sorting requires at least $\Omega(n \lg (pn))$ queries (to succeed with probability at least $2/3$) on an Erd\H{o}s-Renyi random graph $G(n, p)$ (with a planted Hamiltonian path to ensure sorting is possible). Formally, we prove the following theorem:

\begin{theorem} \label{thm:lower_bound}
    Let $n$ and $p \ge \frac{\ln n + \ln \ln n + \omega(1)}{n}$ be known. Suppose that we are given a random graph $G$ created by starting with the Erd\H{o}s-Renyi random graph $G(n, p)$, deciding the true order of the vertices uniformly at random, and then adding the corresponding Hamiltonian path. Then, any algorithm that determines the true order of $G$ with probability at least $2/3$ over the randomness of the algorithm and the randomness of $G$ must use $\Omega(n \log (pn))$ queries.
\end{theorem}

The main technical tool we use is the following bound on the number of Hamiltonian cycles in a random graph, due to Glebov and Krivelevich \cite{glebov2013number}:

\begin{theorem} \label{thm:num_hamilton_cycles}
    For $p \ge \frac{\ln n + \ln \ln n + \omega(1)}{n}$, the number of (undirected) Hamiltonian cycles in the random graph $G(n, p)$ is at least $n! \cdot p^n \cdot (1-o(1))^n$, with probability $1-o(1)$ as $n$ goes to $\infty$.
\end{theorem}

A consequence of this is that, for the graph $G$ in Theorem \ref{thm:lower_bound}, the total number of undirected Hamiltonian \emph{paths} is, with probability $1-o(1)$, also at least $(np \cdot (1-o(1))/e)^{n} = (np)^{\Omega(n)}$.

We will think of the graph $G$ in Theorem \ref{thm:lower_bound} as being constructed through the following process. First, Alice chooses a random permutation $\pi_1, \dots, \pi_n$ of the vertices $V = [n]$ and makes the true ordering $\pi_1 \prec \pi_2 \prec \cdots \prec \pi_n$. Then, she adds the edges $(\pi_1, \pi_2), \ldots, (\pi_{n-1}, \pi_n)$. Finally, for each pair of vertices $i \neq j$ that are not consecutive in the true order, she adds $(i, j)$ to $E$ with probability $p$. We use $G$ to denote the undirected graph that Alice has constructed, and we use $(G, \pi)$ to denote the directed graph. 

Now, suppose an adversary sees the undirected graph $G = (V, E)$ with the planted Hamiltonian path. In addition, suppose the adversary knows some set of directed edges $E'$ (where, as undirected edges, $E' \subset E$). Now, consider a permutation $\sigma$ such that $(\sigma_1, \sigma_2), (\sigma_2, \sigma_3), \cdots, (\sigma_{n-1}, \sigma_n)$ are all undirected edges in $E$, and such that the known orders in $E'$ do not violate $\sigma$. We say that such permutations $\sigma$ are \defn{consistent} with $G$ and $E'$.

The next proposition tells us that the adversary cannot distinguish the true order $\pi$ from any other permutation $\pi'$ that is consistent with $G$ and $E'$. In other words, the adversary must treat all possible Hamiltonian paths as equally viable.

\begin{proposition} \label{prop:uniform}
    Suppose we know $G$ as well as a subset $E'$ of directed edges. Then, the posterior probability of the true ordering of the vertices is uniform over all $\sigma$ that are consistent with $G$ and $E'$.
\end{proposition}

\begin{proof}
    For each permutation $\sigma$, define $\mathcal{S}(\sigma)$ to be the event that Alice chose the permutation $\sigma$. Define $\mathcal{G}$ to be the event that the undirected graph constructed by Alice is the graph $G$, and define $\mathcal{E}'$ be the event that every directed edge $E'$ is in the directed graph $(G, \pi)$. Then, the initial probability that Alice created $(G, \sigma)$ (before we are told $G$ or any of the orders of $E'$) is precisely
    \begin{equation}
        \Pr[\mathcal{S}(\sigma), \mathcal{G}, \mathcal{E}'] = \frac{1}{n!} \cdot p^{|E|-(n-1)} \cdot (1-p)^{(n)(n-1)/2 - |E|}.
        \label{eq:all_paths_equal}
    \end{equation}

    This is because she decided $\pi = \sigma$ with probability $\frac{1}{n!},$ each stochastic edge in $E$ (of which there are $|E|-(n-1)$ of them) is included with probability $p$, and each non-edge vertex pair is excluded with probability $(1-p)$. Finally, given $\sigma, G,$ we know that the event $\mathcal{E}'$ is also true. As a result, even if we condition on $\mathcal{G}$ and $\mathcal{E}'$, this means
\[\Pr[\mathcal{S}(\sigma)\mid\mathcal{G}, \mathcal{E}'] = \frac{\Pr[\mathcal{S}(\sigma), \mathcal{G}, \mathcal{E}']}{\Pr[\mathcal{G}, \mathcal{E}']}\]
    which by \eqref{eq:all_paths_equal} is the same for every $\sigma$ that is consistent with the fixed $G, E'$. Hence, knowing $G$ and $E'$ means each permutation $\sigma$ that is consistent with $G$ and $E'$ is equally likely.
\end{proof}

At this point, the remainder of the proof of the lower bound is relatively simple. The idea is that initially, using Theorem \ref{thm:num_hamilton_cycles}, there are $(np)^{\Omega(n)}$ possible choices for $\sigma$. But each query cannot give more than $1$ bit of information about $\sigma$, so one would need $n \log (np)$ queries to determine $\sigma.$ We now prove this formally.

\begin{proof}[Proof of Theorem \ref{thm:lower_bound}]
    Suppose that $G$ is a graph with $K \ge (np)^{c \cdot n}$ Hamiltonian paths for some fixed constant $c$, which is true with probability at least $1-o(1)$. By Proposition \ref{prop:uniform}, unless we have narrowed down the possible permutations to $1$, we have at most a $1/2$ probability of guessing the true permutation. Therefore, to guess the true permutation with probability $2/3$ overall using $T$ queries, we must be able to perform $T$ queries to reduce the number of consistent permutations $\sigma$ to $1$, with probability at least $1/3-o(1) > 1/4$.
    
    Now, after performing $t \le T$ queries, suppose there are $K_t$ potential consistent permutations. At the beginning, $K_0 = K$, which means $\lg K_0 \ge c n \lg(pn)$. Now, suppose we have $K_t$ possible permutations at some point and we query some edge $e = (u, v)$. Then, for some integers $a, b$ such that $a+b = K_t,$ either we will reduce the number of permutations to $a$ or reduce the number of permutations to $b$, depending on the direction of $e$. Due to the uniform distribution of the consistent edges (by Proposition \ref{prop:uniform}), the first event occurs with probability $\frac{a}{K_t}$ and the second event occurs with probability $\frac{b}{K_t}$. Therefore, the expected value of $\lg K_{t+1}$, assuming we query this edge $e$, is
\[\frac{a}{K_t} \lg a + \frac{b}{K_t} \lg b = \lg K_t + \left(\frac{a}{K_t} \lg \frac{a}{K_t} + \frac{b}{K_t} \lg \frac{b}{K_t}\right) = \lg K_t - H_2(a/K_t),\]
    where $H_2(p) := -[p \lg p + (1-p) \lg (1-p)]$ is known to be at most $1$ for $p \in [0, 1]$. (Note that we are using the convention that $0 \lg 0 = 0$). 
    
    Therefore, no matter what edge we query, the expected value of $\lg K_{t+1}$ is at least $(\lg K_t) - 1.$ Hence, for any $t \ge 0,$ $\mathbb{E}[\lg K_t] \ge (\lg K) - t \ge c n \lg (p n) - t$, regardless of the choice of queries that we make. So, after $T := c n \lg(pn)/4$ queries, the random variable $\lg K_T$ is bounded in the range $[0, \lg K]$ and has expectation at least $\lg K - T \ge \frac{3}{4} \lg K$. Therefore, $\Pr[\lg K_T = 0]$ is at most $\frac{1}{4}$ by Markov's inequality on the random variable $\lg K - \lg K_T$. Since $\lg K_T = 0$ is equivalent to there being exactly one choice for the true permutation $\sigma$, we have that $T = c n \lg(pn)/4$ queries is \textbf{not sufficient} to determine $\sigma$ uniquely with success probability at least $1/3 - o(1) > 1/4$. This concludes the proof.
\end{proof}

\section{Generalized Sorting in Arbitrary (Moderately Sparse) Graphs} \label{sec:Worstcase}

In this section, we improve the worst-case generalized sorting bound of \cite{huang2011algorithms} from $O(\min(n^{3/2} \log n, m))$ to $O(\sqrt{mn}\cdot \log n)$. Hence, for moderately sparse graphs, i.e., $n (\log n)^{O(1)} \le m \le n^2/(\log n)^{O(1)}$, we provide a significant improvement over \cite{huang2011algorithms}.

The main ingredients we use come from the papers \cite{huang2011algorithms}, which provides a convex geometric approach that is especially useful when we are relatively clueless about the true order, and \cite{lu2021generalized}, which shows how to efficiently determine the true order given a sufficiently good approximation to ordering.

Suppose that the original (undirected) graph is some $G = (V, E)$ and that at some point, we have queried some subset $E' \subset E$ of (directed) edges and know their orders. We say that a permutation $\sigma$ is \defn{compatible} with $G$ and $E'$ if for every directed edge $(u, v) \in E'$, $u \prec v$ according to $\sigma$. (Note that being compatible is a slightly weaker property than being consistent, as defined in Section \ref{sec:Lowerbound}, because compatibility does not require that $\sigma$ appear as a Hamiltonian path in $G$.) Let $\Sigma(G, E')$ be the set of permutations that are compatible with $G$ and $E'$. Moreover, for any vertex $u$, define $S_{E', G}(u)$ to be the (unweighted) average of the ranks of $u$ among all $\sigma \in \Sigma(G, E')$. Huang et al. \cite{huang2011algorithms} proved the following theorem using an elegant geometric argument:

\begin{theorem}[Lemma 2.1 of \cite{huang2011algorithms}] \label{thm:convex_geometry}
    Suppose that $S_{E', G}(u) \ge S_{E', G}(v).$ Then, $|\Sigma(G, E' \cup (u, v))| \le \left(1 - \frac{1}{e}\right) \cdot |\Sigma(G, E')|,$ where $(u, v)$ represents the directed edge $u \to v.$
\end{theorem}

The above theorem can be thought of as follows. If we see that $u \prec v$ in the true ordering but previously the expectation of $u$'s rank under all compatible permutations was larger than the expectation of $v$'s rank under all compatible permutations, then the number of compatible permutations decreases by a constant factor.

In addition, we have the following theorem from \cite{lu2021generalized}:

\begin{theorem}[Theorem 1.1 of \cite{lu2021generalized}] \label{thm:predictions}
    Let $\vec{E}$ be the true directed graph of $E$ (with respect to the true permutation $\pi$) and let $\tilde{E}$ be some given directed graph of $E$, such that at most $w$ of the edge directions differ from $\vec{E}$. Then, given the graph $G = (V, E)$, $\tilde{E}$ (but not $\vec{E}$), and $w$, there exists a randomized algorithm \Call{PredictionSort}{$G, \tilde{E}, w$} using $O(w + n \log n)$ queries that can determine $\vec{E}$ with high probability.
\end{theorem}

Our algorithm, which we call \Call{SparseGeneralizedSort}{}, works as follows. Set a parameter $a$ to be $(1 + c)\sqrt{m/n}$ for some sufficiently large positive constant $c$. We repeatedly apply the following procedure until we have found the true ordering.

\begin{enumerate}
    \item Let $E'$ be the current set of queried edges and $\Sigma(G, E')$ be the set of compatible permutations. If $|\Sigma(G, E')| = 1$ then there is only one option for the permutation, which we return.
    \item Else, choose $\sigma$ as the permutation of $[n]$ ordered where $u$ occurs before $v$ if $S_{E', G}(u) \le S_{E', G}(v).$ We break ties arbitrarily.
    \item Let $\tilde{E}$ be our ``approximation'' for the true directed graph $\vec{E}$, where $(u, v) \in \tilde{E}$ if $u$ comes before $v$ in the $\sigma$ order.
    \item Query $a$ random edges from $E$.
    \item If all $a$ edges match the same direction as our approximation $\tilde{E}$, we use the algorithm of Theorem \ref{thm:predictions} with the parameter $w$ set to $\sqrt{m/n} \cdot \log n$. If not, we add the mispredicted edges and their true directions to $E'$, and we return to the first step.
\end{enumerate}

\begin{theorem}
    \Call{SparseGeneralizedSort}{} makes at most $O(\sqrt{mn}\log n)$ queries and successfully identifies the true order of the vertices with high probability.
\end{theorem}

\begin{proof}
    To analyze the algorithm, we must bound the number of queries that it performs, and we must also bound the probability that it fails. Note that there are two ways the algorithm can terminate, either by achieving $|\Sigma(G, E')| = 1$ in Step $1$, or by invoking the algorithm of Theorem \ref{thm:predictions} in Step $5$. (In the latter case, our algorithm terminates regardless of whether the algorithm from Theorem \ref{thm:predictions} successfully finds the true order or not.)

    
    First, we bound the number of queries. Note that the only queries we make are at most the $a$ random queries in each loop plus the $O(n \log n + w)$ queries made by a single application of Theorem \ref{thm:predictions}. Hence, if we repeat this procedure at most $k$ times, we make at most $O(k \cdot a + n \log n + w)$ total queries.
    
    Now, each time we repeat this procedure, this means that we found an edge that went in the opposite direction of our approximation $\tilde{E}$, meaning that we found an edge $u \to v$ such that $S_{E', G}(u) \ge S_{E', G}(v).$ Then, by Theorem \ref{thm:convex_geometry}, $|\Sigma(G, E' \cup (u, v))| \le \left(1 - \frac{1}{e}\right) \cdot |\Sigma(G, E')|.$ Thus, each iteration of the procedure decreases the number of compatible edges by at least a factor of $1 - \frac{1}{e}.$ Since we start off with $n!$ compatible permutations, after $k$ iterations, we have at most $n! \cdot \left(1 - \frac{1}{e}\right)^k$ remaining permutations. So after $k = O(n \log n)$ iterations, we are guaranteed to be down to at most $1$ permutation, ensuring that Step 1 terminates the procedure. Therefore, we make at most $O(n \log n \cdot a + w) = O(\sqrt{m n} \cdot \log n)$ total queries.
    
    Finally, we show that that the algorithm succeeds with high probability. Note that there are two ways that the algorithm could fail: the first is that Step $5$ invokes the algorithm from Theorem \ref{thm:predictions} even though there are more than $w$ mispredicted edges in $\tilde{E}$; and the second is that Step $5$ correctly invokes the algorithm from Theorem \ref{thm:predictions}, but that algorithm fails. By Theorem \ref{thm:predictions}, the latter case happens with low probability, and thus our task is to bound the probability of the first case happening. 
    
    Suppose that for some given $\tilde{E}$, the number of incorrect edges is more than $w = \sqrt{m/n} \cdot \log n$. Then, the probability of a random edge being incorrect is at least $\log n/\sqrt{m/n}$. So, if we sample $a$ edges, at least one of the edges will be wrong with probability $$1 - \left(1 - \frac{\log n}{\sqrt{m/n}}\right)^{a} \ge 1 - e^{- \log n \cdot a/\sqrt{m/n}} = 1 - n^{-(1 + c)}.$$ Therefore, at each time we query $a$ random edges from $E$, with failure probability $n^{-(1 + c)}$, we either find an incorrect direction edge in $\tilde{E}$ or the number of incorrect direction edges in $\tilde{E}$ is at most $w$. Since we only repeat this loop at most $n \log n$ times, the probability that we ever incorrectly invoke the algorithm of Theorem \ref{thm:predictions} is at most $(n \log n)/n^{(1 + c)} = 1/\poly(n)$.
\end{proof}

\begin{remark}
    Unlike the stochastic sorting algorithm, it is not immediately clear how to perform \Call{SparseGeneralizedSort}{} in polynomial time, since the average rank $S_{E', G}(v)$ is difficult to compute. We briefly note, however, that $S_{E', G}(v)$ can be approximated with $\pm O(1)$ error in polynomial time by sampling from the polytope in $\mathbb{R}^n$ with facets created by $x_i \le x_j$ if we know that $i \prec j$, for each pair $(i, j)$ (see Lemma 3.1 of \cite{huang2011algorithms} as well as discussion in \cite{dyer1991sampling} for more details). Moreover, these approximations suffice for our purposes, since it is known (see Lemma 2.1 of \cite{huang2011algorithms}) that Theorem \ref{thm:convex_geometry} still holds even if $S_{E', G}(u) \ge S_{E', G}(v) - O(1)$ (although the modified theorem now has a different constant than $\left(1 - \frac{1}{e}\right)$).
\end{remark}

\section*{Acknowledgments}

S. Narayanan is funded by an NSF GRFP Fellowship and a Simons Investigator Award. W. Kuszmaul is funded by a Fannie \& John Hertz Foundation Fellowship; by an NSF GRFP Fellowship; and by the United States Air Force Research Laboratory and the United States Air Force Artificial Intelligence Accelerator and was accomplished under Cooperative Agreement Number FA8750-19-2-1000. The views and conclusions contained in this document are those of the authors and should not be interpreted as representing the official policies, either expressed or implied, of the United States Air Force or the U.S. Government. The U.S. Government is authorized to reproduce and distribute reprints for Government purposes notwithstanding any copyright notation herein

The authors would like to thank Nicole Wein for several helpful conversations during the start of this project.

\bibliography{references}

\appendix

\clearpage
\section{Pseudocode} \label{sec:Pseudocode}

In this Appendix, we provide pseudocode for both of our algorithms. Algorithms \ref{Findx1}, \ref{CreateLevel}, \ref{Increment}, \ref{Find}, and \ref{StochasticSort} comprise the \Call{StochasticSort} algorithm, with Algorithm \ref{StochasticSort} our main routine. Algorithm \ref{GeneralSort} implements the \Call{SparsGeneralizedSort} algorithm.

As a notation, we will use $\Call{Query}{u, v}$ to be the comparison function that takes two vertices $u, v$ connected by an edge, and returns $1$ if $u \prec v$ and $0$ if $v \prec u$.

\begin{algorithm}
\caption{}
\begin{algorithmic}[1]
\Procedure{Findx1}{} \Comment{Find the first vertex $x_1$}
\State $v_0 = $ NULL, $S = \{\}$
\State Pick arbitrary edge $(u, v)$
\If{$\Call{Query}{u, v} = 1$} \Comment{Means that $u$ comes before $v$}
    \State $v_0 \leftarrow u,$ $S \leftarrow \{v\}$
\Else
    \State $v_0 \leftarrow v,$ $S \leftarrow \{u\}$
\EndIf
\While{$\exists u$ such that $(u, v_0) \in E$ and $u \not\in S$}
    \If{$\Call{Query}{u, v_0} = 1$}
        \State $S \leftarrow S \cup \{v_0\}$
        \State $v_0 \leftarrow u$
    \Else
        \State $S \leftarrow S \cup \{u\}$
    \EndIf
\EndWhile
\State \textbf{Return} $v_0$
\EndProcedure
\end{algorithmic}
\label{Findx1}
\end{algorithm}

\begin{algorithm}
\caption{}
\begin{algorithmic}[1]
\Procedure{CreateLevel}{$i$} \Comment{Create level $L_i$ and all lower levels}
\State $L_i \leftarrow L_{i+1}$
\For{each $v$ in $L_{i+1}$}
    \State elim[$v$] $\leftarrow$ NULL \Comment{Recreating levels, so $v$ may become unblocked}
    \For{each $u$ in $L_{i+c}$}
        \If{$(u, v) \in E_i$}
            \If{$\Call{Query}{u, v} = 1$} \Comment{Means that $u$ comes before $v$}
                \State $L_i \leftarrow L_i \backslash \{v\}$ \Comment{Remove $v$ from $L_i$, since it is blocked by $u$}
                \State elim[$v$] $\leftarrow$ $u$ \Comment{keep track of which vertex blocked $v$}
            \EndIf
        \EndIf
    \EndFor
\EndFor
\If{$i > 1$}
    \State \Call{CreateLevel}{$i-1$} \Comment{recursively do lower levels}
\EndIf
\EndProcedure
\end{algorithmic}
\label{CreateLevel}
\end{algorithm}

\begin{algorithm}
\caption{}
\begin{algorithmic}[1]
\Procedure{Increment}{$\ell$} \Comment{Do after current lowest vertex $x_\ell$ found}
\For{$i = 1$ to $q+c$}
    \State $L_i = L_i \backslash \{x_\ell\}$ \Comment{Remove $x_\ell$ from all levels}
\EndFor
\For{each $v$ such that elim[$v$]$= x_\ell$}
    \State elim[$v$] $\leftarrow$ NULL \Comment{$x_{\ell}$ is discovered, so $v$ no longer blocked by $x_{\ell}$}
    \State Let $i$ be the lowest level containing $v$
    \While{$i > 1$}
        \For{each $u$ in $L_{i+c}$}
            \If{$(u, v) \in E_i$}
                \If{$\Call{Query}{u, v} = 1$} \Comment{Means that $u$ comes before $v$}
                    \State elim[$v$] $\leftarrow$ $u$ \Comment{keep track of which vertex eliminated $v$}
                    \State \textbf{break} \Comment{End procedure for $v$}
                \EndIf
            \EndIf
        \EndFor
        \State $i \leftarrow i-1$ \Comment{Decrement level of $v$ if $v$ isn't blocked by any $u$}
        \State $L_{i-1} \leftarrow L_{i-1} \cup \{v\}$ \Comment{Add $v$ to level $L_{i-1}$}
    \EndWhile
\EndFor
\EndProcedure
\end{algorithmic} 
\label{Increment}
\end{algorithm}

\begin{algorithm}
\caption{}
\begin{algorithmic}[1]
\Procedure{Find}{$\ell$} \Comment{Find the next vertex in the order, $x_{\ell}$} 
\State $S = \emptyset$ \Comment{$S$ is the set of candidates that could be $x_{\ell}$}
\For{each $v \not\in \{x_1, \dots, x_{\ell-1}\}$}
    \If{$v \in L_1$ \textbf{and} $(x_{\ell-1}, v) \in E$}
        \State $S \leftarrow S \cup \{v\}$ \Comment{Add $v$ to set of candidates}
    \EndIf
\EndFor

\For{$i = 1$ to $q$}
    \For{each $v$ in $S$}
        \For{each $u$ in $L_i$} \Comment{Check if something in $i$th level has edge to $v$}
            \If{$(u, v) \in E$}
                \If{$\Call{Query}{u, v} = 1$}
                    \State $S = S \backslash \{v\}$ \Comment{$v$ is no longer a candidate}
                \EndIf
            \EndIf
            \If{$|S| = 1$} \Comment{once we've reached a single candidate}
                \State \textbf{Return} $S[0]$ \Comment{Return the only element in $S$, and break}
            \EndIf
        \EndFor
    \EndFor
\EndFor
\EndProcedure
\end{algorithmic}
\label{Find}
\end{algorithm}

\begin{algorithm}
\caption{}
\begin{algorithmic}[1]
\Procedure{StochasticSort}{$G = (V, E)$} \Comment{Recover $x_1, \dots, x_{n/2}.$ $n, p$ are known} 
\State Initialize all levels $L_1, \dots, L_{q+c}$ as $[n]$
\State Initialize elim[$v$] = NULL for all $v \in [n]$
\State \Call{CreateLevel}{$q$}
\For{$\ell = 1$ to $n/2$} \Comment{Recover $x_\ell$ in order}
    \If{$\ell = 1$}
        \State $x_\ell =$ \Call{Findx1}{} \Comment{Use different procedure for finding first vertex}
    \Else
        \State $x_\ell =$ \Call{Find}{$\ell$}
    \EndIf
    \State \Call{Increment}{$\ell$}
    \If{$\ell = \lfloor \ell' \cdot p^{-1}/16 \rfloor$ for some integer $\ell'$}
        \State \Call{CreateLevel}{$1+\nu_2(\ell')$} \Comment{$\nu_2(\ell'):$ largest exponent of $2$ dividing $\ell'$}
    \EndIf
\EndFor
\EndProcedure
\end{algorithmic}
\label{StochasticSort}
\end{algorithm}

\begin{algorithm}
\caption{}
\begin{algorithmic}[1]
\Procedure{SparseGeneralSort}{$G = (V, E)$} \Comment{Recover sorted order for arbitrary graph. $n = |V|, m = |E|$.}
\State $a \leftarrow 2 \sqrt{m/n}$, $w \leftarrow \sqrt{m/n} \cdot \log n$
\State $E' = \emptyset$ \Comment{$E'$ is current set of queried edges, currently empty}
\While{$|\Sigma(G, E')| \ge 2$, where $\Sigma(G, E')$ is the set of permutations that do not violate any of the edge orders in $E'$ } 
    \For{each vertex $v$}
        \State $S_{E', G}(v) :=$ average rank of $v$ among all permutations in $\Sigma(G, E')$.
    \EndFor
    \State $\tilde{E} \leftarrow \emptyset$ \Comment{$\tilde{E}$ will be the predicted edge orders}
    \For{each edge $(u, v) \in E$}
        \If{$S_{E', G}(u) \le S_{E', G}(v)$}
            \State $\tilde{E} \leftarrow \tilde{E} \cup (u, v)$
        \Else
            \State $\tilde{E} \leftarrow \tilde{E} \cup (v, u)$
        \EndIf
    \EndFor
    \State \textbf{boolean} order $\leftarrow$ TRUE \Comment{Checks if the $a$ random edges are predicted correctly by $\tilde{E}$}
    \For{$i$ from $1$ to $a$}
        \State $(u_i, v_i) \sim E$ uniformly at random \Comment{WLOG $(u_i, v_i)$ is ordered such that $(u_i, v_i) \in \tilde{E}$}

        \If{$\Call{Query}{u_i, v_i} = 0$} \Comment{Means that $v_i \prec u_i$ in the true order}
            \State $E' \leftarrow E' \cup (v_i, u_i)$
            \State order $\leftarrow$ FALSE
        \EndIf
    \EndFor
    \If{order $=$ TRUE}
        \State \textbf{Return} \Call{PredictionSort}{$G, \tilde{E}, w$}
    \EndIf
\EndWhile
\State \textbf{Return} the unique permutation in $\Sigma(G, E')$
\EndProcedure
\end{algorithmic} 
\label{GeneralSort}
\end{algorithm}


\end{document}